\documentclass[11pt,oneside,english]{amsart}
\usepackage[T1]{fontenc}
\usepackage[latin9]{inputenc}
\usepackage[a4paper]{geometry}
\usepackage{amstext}
\usepackage{amsthm}
\usepackage{amssymb}

\makeatletter
\numberwithin{equation}{section}
\numberwithin{figure}{section}
\theoremstyle{plain}
\newtheorem{thm}{\protect\theoremname}
\theoremstyle{plain}
\newtheorem{prop}[thm]{\protect\propositionname}
\theoremstyle{remark}
\newtheorem{rem}[thm]{\protect\remarkname}
\theoremstyle{plain}
\newtheorem{cor}[thm]{\protect\corollaryname}
\theoremstyle{plain}
\newtheorem{lem}[thm]{\protect\lemmaname}

\allowdisplaybreaks

\makeatother

\usepackage{babel}
\providecommand{\corollaryname}{Corollary}
\providecommand{\lemmaname}{Lemma}
\providecommand{\propositionname}{Proposition}
\providecommand{\remarkname}{Remark}
\providecommand{\theoremname}{Theorem}

\begin{document}
\title{A Decomposition Formula for Fractional Heston Jump Diffusion Models}
\date{\today }
\author{Marc Lagunas-Merino}
\address{(\emph{Marc Lagunas-Merino}) Department of Mathematics, University
of Oslo\\
P.O. Box 1053, Blindern\\
N\textendash 0316 Oslo, Norway}
\email{marclagu@math.uio.no}
\author{Salvador Ortiz-latorre}
\address{(\emph{Salvador Ortiz-Latorre}) Department of Mathematics, University
of Oslo\\
P.O. Box 1053, Blindern\\
N\textendash 0316 Oslo, Norway}
\email{salvadoo@math.uio.no}
\thanks{We are grateful for the financial support from Department of Mathematics,
University of Oslo.}
\begin{abstract}
We present an option pricing formula for European options in a stochastic
volatility model. In particular, the volatility process is defined
using a fractional integral of a diffusion process and both the stock
price and the volatility processes have jumps in order to capture
the market effect known as leverage effect. We show how to compute
a martingale representation for the volatility process. Finally, using
It\^{o} calculus for processes with discontinuous trajectories, we develop
a first order approximation formula for option prices. There are two
main advantages in the usage of such approximating formulas to traditional
pricing methods. First, to improve computational efficiency, and second,
to have a deeper understanding of the option price changes in terms
of changes in the model parameters. 
\end{abstract}

\maketitle

\section{Introduction}

Classical stochastic volatility models, where the volatility also
follows a diffusion process, have been proven to be capable of reproducing
some important features of the implied volatility together with its
variation with respect to the strike price. These features are usually
described through the smile or skew, see \cite{Gatheral06}. One of
the main downsides from these class of models is their inability to
explain what is known as the term structure of the skew, i.e. the
dependence of the skew to the time to maturity.

For instance, it can be observed in \cite{GaJaRos18}, how a decrease
of the smile amplitude when time to maturity increases, turns out
to be much slower than it should be according to standard stochastic
volatility models. On the other hand, the observed short-time implied
volatility skew slope tends to infinity as time to maturity tends
to zero, whereas this limit is a constant under classical stochastic
volatility models.

On one hand, the long-memory features for the volatility process can
be achieved by the introduction of fractional noises with a Hurst
parameter $H>1/2$ in the volatility process, as introduced by Comte
and Renault in \cite{ComRen98} and deeply studied in \cite{AloYan17}.
This allows to endow the volatility with high persistence in the long-run,
showing the steepness of long-term volatility smiles without over
increasing the short-run persistence. On the other hand, it was proved
in \cite{AloLeoVives07} that these models fail to describe the short-time
behavior of implied volatility. In order to overcome this limitation,
we present one of the possible approaches, consisting in adding a
jump term to both the stock price and the volatility processes with
a correlation factor between both jumps in order to reproduce the
so-called leverage effect\footnote{This is a well known effect observed in most markets that shows how
most measures of volatility of an asset are negatively correlated
with the returns of that asset.}. This results into a combination between the fractional setup of
Al\`os and Yang \cite{AloYan17} and the jump diffusion framework, first
studied by Merton in \cite{Merton76}. In the present work, we build
an option price approximation methodology for a combination of the
fractional and the jump diffusion models previously mentioned, which
capture the short-time and long-time behavior of implied volatility.
Hence, we will study what we will call fractional stochastic volatility
jump diffusion (FSVJJ) models with both jumps in the price and volatility
processes.

The first stochastic volatility models with jumps were introduced
by Bates in \cite{Bates96} and were achieved by incorporating jumps
to stochastic variance processes, previously introduced by Heston
(1993) in \cite{Heston93}. In our case of study, the variance of
stock prices follows a combination of a CIR process \cite{CoxIngRoss85},
a fractional integral of the stochastic term in the CIR process and
a jump term driven by a L\'evy-type process. Adding the jump framework
to the model should improve the market fit for short-term maturity
options, overcoming the original problem of the Heston stochastic
volatility model. The last would require unrealistically high values
for the vol-of-vol parameter in order to obtain a reasonable fit of
short-term smiles. An alternative approach to model short-term smiles
is the use of rough fractional volatility models, see for instance
\cite{GaJaRos18,EuFuMaGaRo19}. 

Pricing derivatives under stochastic volatility jump diffusion models
involves, naturally, an extra degree of complexity compared to the
standard Black-Scholes pricing framework. This has motivated the development
of approximating formulas in the literature such as \cite{HullWhite87,Alos2006,Alos2012,GLMV19,MePoSoVi18}.
These formulas provide good intuition on the behavior of the smiles
and a better understanding of the effects of changes in the model
parameters onto the price of a derivative. Despite not being closed
pricing formulae, they bring clarity to the practitioner to understand
the effects of model parameters in the option price. As well as, speed
up the calibration process as proved in \cite{GLMV19}. We will use
this idea to find a general decomposition formula for a fractional
Heston model with jumps in the price and the volatility processes
under basic integrability conditions. In a recent paper \cite{merino2019decomposition},
the authors find decomposition formulas in the setup of rough volatility
models. 

This paper is organized as follows. First, we introduce in Section
\ref{sec: FSVJJ Model} a detailed description of the model that will
be used throughout the article. Section \ref{sec: Notation} is devoted
to introduce some preliminary concepts and notation needed in later
parts of our study. Then, in Section \ref{sec: General-expansion-formulas}
we present an exact expansion formula for option pricing in terms
of the Black-Scholes formula adjusted by extra terms that depend on
the future expected volatility, the jumps and correlation parameters.
We provide a martingale representation of future expected volatility
in Section \ref{sec: Martingale_Representation}, by means of the
Clark-Ocone-Haussmann formula and end the section providing its dynamics.
We conclude by developing a first order approximating formula of a
call option under our fractional Heston jump diffusion model. Finally,
we provide an appendix in Section \ref{sec:Appendix} with some additional
technical results.

\section{\label{sec: FSVJJ Model}A Fractional Heston Model with Jumps in
the Stock Log-Prices and Volatility (FSVJJ) Model}

Let $T>0$ be fixed time horizon and let $\left(\Omega,\mathcal{F},\mathbb{Q}\right)$
be a complete probability space. Assume that in $\left(\Omega,\mathcal{F},\mathbb{Q}\right)$
there are defined $W=\left\{ W\right\} _{t\in\left[0,T\right]}$ and
$\tilde{W}=\left\{ \tilde{W}\right\} _{t\in\left[0,T\right]}$, two
independent standard Brownian motions. Moreover, assume that in $\left(\Omega,\mathcal{F},\mathbb{Q}\right)$
there is defined a L\'evy subordinator\footnote{Subordinators are L\'evy processes with increasing paths. Alternatively,
they are L\'evy processes with finite variation paths and postive jumps.
These type of processes are often used in L\'evy-based financial models.
See \cite{ContTankov04}. } $J=\left\{ J\right\} _{t\in\left[0,T\right]}$ which is independent
of the Brownian motions $W$ and $\tilde{W}$. We define the filtration
$\mathbb{F}=\left\{ \mathcal{F}_{t}\right\} _{t\in\left[0,T\right]}$
to be the minimal augmented filtration generated by $W,\tilde{W}$
and $J$. We assume that $J$ is a L\'evy subordinator with generating
triplet $\left(0,\ell,\gamma\right)$, that is,
\begin{align*}
\mathbb{E}\left[e^{iyJ_{t}}\right] & =\exp\left(t\left(i\gamma y+\int_{0}^{\infty}\left(e^{iyz}-1-z\mathbf{1}_{\left\{ 0<z<\infty\right\} }\right)\ell(dz)\right)\right)\\
 & =\exp\left(t\left(iby+\int_{0}^{\infty}\left(e^{iyz}-1\right)\ell(dz)\right)\right),
\end{align*}
with $y\in\mathbb{R},\text{\ensuremath{b=\gamma-\int_{0}^{1}z\ell(dz)\geq0}},$
and $\ell$ a L\'evy measure with support on $\left(0,\infty\right)$
and satisfying $\int_{0}^{1}z\ell\left(dz\right)<\infty$. Note that
we are not assuming $\ell\left(\left(0,\infty\right)\right)<\infty$
and, therefore, the process $J$ may have infinite activity, that
is, the $J$ may have an infinite number of small jumps on any finite
time interval. We will assume that for some $C>0$ we have $\int_{1}^{\infty}e^{Cz}\ell(dz)<\infty.$
Then $J$ has moments of all orders. In particular, $J$ has finite
expectation which is equivalent to $\int_{1}^{\infty}z\ell(dz)<\infty$,
see Proposition 3.13 in \cite{ContTankov04}. By the L\'evy-It\^{o} decomposition
we have that $J$ can be written as 
\begin{align*}
J_{t} & =\gamma t+\int_{0}^{t}\int_{0}^{1}z\tilde{N}\left(ds,dz\right)+\int_{0}^{t}\int_{1}^{\infty}zN\left(ds,dz\right)\\
 & =\gamma t-\int_{0}^{t}\int_{0}^{1}z\ell\left(dz\right)ds+\int_{0}^{t}\int_{0}^{\infty}zN\left(ds,dz\right)\\
 & =bt+\int_{0}^{t}\int_{0}^{\infty}zN\left(ds,dz\right)\\
 & =\left(b+\int_{0}^{\infty}z\ell\left(dz\right)\right)t+\int_{0}^{t}\int_{0}^{\infty}z\tilde{N}\left(ds,dz\right)\\
 & =\left(b+\int_{0}^{\infty}z\ell\left(dz\right)\right)t+\tilde{J}_{t}\\
 & =\left(\gamma+\int_{1}^{\infty}z\ell(dz)\right)t+\tilde{J}_{t},
\end{align*}
where $N\left(ds,dz\right)$ denotes the Poisson random measure with
L\'evy measure $\ell$ and $\tilde{N}\left(ds,dz\right)\triangleq N\left(ds,dz\right)-\ell\left(dz\right)ds$
denotes its associated compensated Poisson random measure. 

Next we introduce the volatility process that we will use in our model.
Following \cite{AloYan17}, we start by considering a CIR process
$\bar{\sigma}^{2}=\left\{ \bar{\sigma}_{t}^{2}\right\} _{t\in\left[0,T\right]}$
of the following form
\begin{equation}
d\bar{\sigma}_{t}^{2}=\kappa\left(\theta-\bar{\sigma}_{t}^{2}\right)dt+\nu\sqrt{\bar{\sigma}_{t}^{2}}dW_{t},\label{eq: CIRprocess_dynamics}
\end{equation}
with $\bar{\sigma}_{0}^{2},\theta,\kappa,\nu>0$ and satisfying the
so called Feller condition \cite{Feller51} $2\kappa\theta\geq\nu^{2}$,
in order to ensure positivity for the variance process, see \cite{EuFuMaGaRo19}.
Applying It\^{o} formula to the process $e^{\kappa t}\bar{\sigma}_{t}^{2}$
one can easily see that the following holds
\begin{equation}
\bar{\sigma}_{t}^{2}=\theta+e^{-\kappa t}\left(\bar{\sigma}_{0}^{2}-\theta\right)+\nu\int_{0}^{t}e^{-\kappa(t-s)}\sqrt{\bar{\sigma}_{s}^{2}}dW_{s}.\label{eq: CIRprocess_integralversion}
\end{equation}
Note that we can write $\bar{\sigma}^{2}$ as the sum of two processes
$Y=\left\{ Y_{t}\right\} _{t\in\left[0,T\right]}$ and $Z=\left\{ Z_{t}\right\} _{t\in\left[0,T\right]},$
defined by
\begin{align*}
U_{t} & \triangleq\theta+e^{-\kappa t}\left(\bar{\sigma}_{0}^{2}-\theta\right),\qquad t\in\left[0,T\right],
\end{align*}
and
\[
Z_{t}\triangleq\int_{0}^{t}e^{-\kappa(t-s)}\sqrt{\bar{\sigma}_{s}^{2}}dW_{s},\qquad t\in\left[0,T\right].
\]
Finally, we will consider a fractional volatility model with only
positive jumps built on a combination of the processes $Y,Z$ and
a fractional integral of $Z$. Let us denote this fractional volatility
process with jumps by $\sigma^{2}=\left\{ \sigma_{t}^{2}\right\} _{t\in\left[0,T\right]}$,
where
\begin{equation}
\sigma_{t}^{2}\triangleq U_{t}+c_{1}\nu Z_{t}+c_{2}\nu I_{0+}^{H-\frac{1}{2}}Z_{t}+c_{3}\eta J_{t},\qquad t\in\left[0,T\right],\label{eq: fHestonJump_Volatility}
\end{equation}
where $H\in\left(\frac{1}{2},1\right)$, $c_{1},c_{2},c_{3},\eta\geq0$
and $I_{0+}^{H-\frac{1}{2}}$ is the left-sided fractional Riemann-Liouville
integral of the path $Z_{\cdot}\left(\omega\right)$ of order $H-\frac{1}{2}$
on $\left[0,T\right]$. Recall that, given $f\in L^{1}\left(\left[0,T\right]\right)$,
the left-sided fractional Riemann-Liouville integral of $f$ of order
$\alpha\in\mathbb{R}_{+}$ on $\left[0,T\right]$ is defined as
\begin{equation}
\left(I_{0+}^{\alpha}f\right)\left(t\right)\triangleq\frac{1}{\Gamma\left(\alpha\right)}\int_{0}^{t}f\left(u\right)\left(t-u\right)^{\alpha-1}du.\label{eq: Fractional_Integral}
\end{equation}
On the other hand let $\rho_{1}\in\left[-1,1\right]$, $\rho_{2}\leq0$
and consider a L\'evy model for the dynamics of the stock log-price
in the time interval $\left[0,T\right]$ given by the following equation
\[
X_{t}=x+rt-\frac{1}{2}\int_{0}^{t}\sigma_{s}^{2}ds+\int_{0}^{t}\sigma_{s}\left(\rho_{1}dW_{s}+\sqrt{1-\rho_{1}^{2}}dB_{s}\right)+\rho_{2}\eta J_{t},
\]
where $r$ is the risk-free rate. For the process $e^{-rt}e^{X_{t}}$
to be a local martingale, we must assume, see Corollary 5.2.2 in \cite{Applebaum09},
that
\[
\rho_{2}\eta\gamma+\int_{0}^{1}\left(e^{\rho_{2}\eta z}-1-\rho_{2}\eta z\mathbf{1}_{\left\{ 0<z<1\right\} }\right)\ell\left(dz\right)+\int_{1}^{\infty}\left(e^{\rho_{2}\eta z}-1\right)\ell\left(dz\right)=0.
\]
Note that $\gamma=-\frac{1}{\rho_{2}\eta}\int_{0}^{\infty}\left(e^{\rho_{2}\eta z}-1-\rho_{2}\eta z\mathbf{1}_{\left\{ 0<z<1\right\} }\right)\ell\left(dz\right)$.
Since we have already proved that $J_{t}=\left(\gamma+\int_{1}^{\infty}z\ell(dz)\right)t+\tilde{J}_{t}$,
it is trivial to see that
\begin{align*}
\rho_{2}\eta J_{t} & =\left(-\int_{0}^{\infty}\left(e^{\rho_{2}\eta z}-1-\rho_{2}\eta z\mathbf{1}_{\left\{ 0<z<1\right\} }\right)\ell\left(dz\right)+\rho_{2}\eta\int_{1}^{\infty}z\ell(dz)\right)t+\rho_{2}\eta\tilde{J}_{t}\\
 & =-\int_{0}^{\infty}\left(e^{\rho_{2}\eta z}-1-\rho_{2}\eta z\right)\ell\left(dz\right)t+\rho_{2}\eta\tilde{J}_{t}\\
 & \triangleq-\zeta\left(\rho_{2},\eta\right)t+\rho_{2}\eta\tilde{J}_{t}
\end{align*}

Now we have that, whenever a jump occurs in the volatility process,
since $\rho_{2}\leq0$ we get a negative jump in log-prices, so we
can model the leverage effect. We can also specify the model through
the following equations
\begin{align}
X_{t} & =x+\left(r-\zeta\left(\rho_{2},\eta\right)\right)t-\frac{1}{2}\int_{0}^{t}\sigma_{s}^{2}ds+\int_{0}^{t}\sigma_{s}\left(\rho_{1}dW_{s}+\sqrt{1-\rho_{1}^{2}}d\tilde{W}_{s}\right)+\rho_{2}\eta\tilde{J}_{t},\label{eq: log-price_SDE}\\
\sigma_{t}^{2} & =Y_{t}+c_{1}\nu Z_{t}+c_{2}\nu I_{0+}^{H-\frac{1}{2}}Z_{t}+c_{3}\eta J_{t}.
\end{align}

\section{\label{sec: Notation}Preliminaries and notation}

Following similar ideas to the ones found in \cite{AloYan17}, we
will extend the decomposition formula to a fractional Heston model
with infinite activity jumps in both prices and volatility. It is
well known that $V_{t}$, the value at time $t$ of a derivative whose
payoff is $h\left(X_{T}\right)$, is given by the risk neutral pricing
formula
\[
V_{t}\left(h\right)=e^{-r(T-t)}\mathbb{E}\left[h\left(X_{T}\right)\mid\mathcal{F}_{t}\right].
\]
We now proceed to introduce some definitions and notations which will
be used throughout the paper:
\begin{itemize}
\item We will denote $\mathbb{E}_{t}\left[\cdot\right]\triangleq\mathbb{E}\left[\cdot\mid\mathcal{F}_{t}\right].$
\item Let $BS\left(t,x,\sigma\right)$ denote the price of a plain vanilla
European call option under the classical Black-Scholes pricing formula
with constant volatility $\sigma$, stock log-price $x$, strike price
$K$, time to maturity $T-t$, and constant interest rate $r$. In
this case,
\[
BS\left(t,x,\sigma\right)=e^{x}\Phi\left(d_{+}\right)-Ke^{-r(T-t)}\Phi\left(d_{-}\right),
\]
where $\Phi$ denotes the standard normal cumulative probability function
and $d_{\pm}$ is defined as 
\[
d_{\pm}\triangleq\frac{x-\ln K-r(T-t)}{\sigma\sqrt{T-t}}\pm\frac{\sigma}{2}\sqrt{T-t}.
\]
\item In our setting, the price of a call option at time $t$ is given by
\[
V_{t}=e^{-r\left(T-t\right)}\mathbb{E}_{t}\left[\left(e^{X_{T}}-K\right)^{+}\right].
\]
\item We recall from the Feynman-Kac formula for the continuous version
of the model $\left(\ref{eq: log-price_SDE}\right)$, the operator
\[
\mathcal{L}_{\sigma}\triangleq\partial_{t}+\frac{1}{2}\sigma^{2}\partial_{xx}^{2}+\left(r-\frac{1}{2}\sigma^{2}\right)\partial_{x}-r.
\]
Note that $\mathcal{L}_{\sigma}BS\left(\cdot,\cdot,\sigma\right)=0$
by construction.
\item We will use an adapted projection of the future average variance defined
by
\begin{equation}
v_{t}^{2}\triangleq\frac{1}{T-t}\int_{t}^{T}\mathbb{E}_{t}\left[\sigma_{s}^{2}\right]ds,\label{eq: Projected_Future_Avg_Variance}
\end{equation}
to obtain a decomposition of $V_{t}$ in terms of $v_{t}$. This idea,
used in \cite{Alos2012}, switches an anticipative problem into a
non-anticipative one, related to the adapted process $v_{t}$.
\item We define $M_{t}\triangleq\int_{0}^{T}\mathbb{E}_{t}\left[\sigma_{s}^{2}\right]ds$.
Notice then that the projected future average variance can be written
as $v_{t}^{2}=\frac{1}{T-t}\left(M_{t}-\int_{0}^{t}\sigma_{s}^{2}ds\right).$
Recall that, by definition, $M$ is a martingale with respect to the
filtration generated by $W$ and $J$, is also $\mathcal{F}^{\tilde{W}}$-independent
and its dynamics is given by
\[
dM_{t}=\nu A(T,t)\sqrt{\bar{\sigma}_{t}^{2}}dW_{t}+c_{3}\eta dJ_{t},
\]
as it is later proved in Proposition \ref{prop: Martingale_Dynamics}.
It will also be useful to introduce $M^{c}$, the continuous part
of the process $M$, with dynamics given by 
\[
dM_{t}^{c}=\nu A(T,t)\sqrt{\bar{\sigma}_{t}^{2}}dW_{t}.
\]
\item Let $\left\{ X_{t}\right\} _{t\in\left[0,T\right]}$ and $\left\{ Y_{t}\right\} _{t\in\left[0,T\right]}$
be two It\^{o}-L\'evy processes given by the following dynamics
\begin{align*}
dX_{t} & =\alpha_{x}\left(t\right)dt+\beta_{x}\left(t\right)dW_{t}+\int_{0}^{\infty}\gamma_{x}\left(t,z\right)\tilde{N}\left(dt,dz\right),\\
dY_{t} & =\alpha_{y}\left(t\right)dt+\beta_{y}\left(t\right)dW_{t}+\int_{0}^{\infty}\gamma_{y}\left(t,z\right)\tilde{N}\left(dt,dz\right).
\end{align*}
Given a function $F\in C^{0,1,1}\left(\left[0,T\right]\times\mathbb{R}\times\mathbb{R}\right)$,
we define
\begin{align*}
\Delta_{x}F\left(t,X_{t-},Y_{t-}\right) & \triangleq F\left(t,X_{t-}+\gamma_{x}\left(t,z\right),Y_{t-}\right)-F\left(t,X_{t-},Y_{t-}\right),\\
\Delta_{y}F\left(t,X_{t-},Y_{t-}\right) & \triangleq F\left(t,X_{t-},Y_{t-}+\gamma_{y}\left(t,z\right)\right)-F\left(t,X_{t-},Y_{t-}\right),\\
\Delta_{x}^{2}F\left(t,X_{t-},Y_{t-}\right) & \triangleq\Delta_{x}F\left(t,X_{t-},Y_{t-}\right)-\gamma_{x}\left(t,z\right)\left(\partial_{x}F\right)\left(t,X_{t-},Y_{t-}\right),\\
\Delta_{y}^{2}F\left(t,X_{t-},Y_{t-}\right) & \triangleq\Delta_{y}F\left(t,X_{t-},Y_{t-}\right)-\gamma_{y}\left(t,z\right)\left(\partial_{y}F\right)\left(t,X_{t-},Y_{t-}\right).
\end{align*}
\item The introduction of the following differential operators is very convenient
for notational purposes, and both expressions will be used indistinctly
throughout the article. 
\begin{align*}
\Lambda & \triangleq\partial_{x},\\
\Gamma & \triangleq\left(\partial_{x}^{2}-\partial_{x}\right);\qquad\Gamma^{2}=\Gamma\circ\Gamma=\left(\partial_{x}^{4}-2\partial_{x}^{3}+\partial_{x}^{2}\right).
\end{align*}
\item Given two continuous semimartingales $X$ and $Y$, we define the
following processes
\begin{align*}
L\left[X,Y\right]_{t} & \triangleq\mathbb{E}_{t}\left[\int_{t}^{T}\sigma_{u}d\left[X,Y\right]_{u}\right],\\
D\left[X,Y\right]_{t} & \triangleq\mathbb{E}_{t}\left[\int_{t}^{T}d\left[X,Y\right]_{u}\right],
\end{align*}
for $t\in\left[0,T\right].$
\end{itemize}
Since the derivatives of $BS\left(t,X_{t},v_{t}\right)$ are not bounded,
we will make use of an approximating argument. Consider the approximation
$v_{t}^{\delta}$ of $v_{t}$ for a fixed $\delta>0$, given by 
\begin{equation}
v_{t}^{\delta}\triangleq\sqrt{\frac{1}{T-t}\left(\delta+\int_{t}^{T}\mathbb{E}_{t}\left[\sigma_{s}^{2}\right]ds\right)}=\sqrt{\frac{1}{T-t}\left(\delta+M_{t}-\int_{0}^{t}\sigma_{s}^{2}ds\right)}.\label{eq:V_delta}
\end{equation}
The following proposition shows how the dynamics of this process is
obtained.
\begin{prop}
For fixed $\delta>0$, let $v^{\delta}$ be the approximation of the
adapted projection of the average future volatility process, defined
in $\left(\ref{eq:V_delta}\right)$. Then, the dynamics of $v^{\delta}$
is given by
\begin{align}
dv_{t}^{\delta} & =\frac{(v_{t}^{\delta})^{2}-\sigma_{t}^{2}}{2v_{t}^{\delta}(T-t)}dt+\frac{dM_{t}^{c}}{2v_{t}^{\delta}(T-t)}-\frac{d\left[M^{c},M^{c}\right]_{t}}{8\left(v_{t}^{\delta}\right)^{3}\left(T-t\right)^{2}}\label{eq: future_avg_vol_dynamics}\\
 & \qquad+\int_{0}^{\infty}\Delta_{m}^{2}g^{\delta}\left(t,M_{t-},Y_{t-}\right)\ell\left(dz\right)dt+\int_{0}^{\infty}\Delta_{m}g^{\delta}\left(t,M_{t-},Y_{t-}\right)\tilde{N}\left(dt,dz\right),\nonumber 
\end{align}
where $g^{\delta}\left(t,m,y\right)=\sqrt{\frac{1}{T-t}\left(\delta+m-y\right)}.$
\end{prop}

\begin{proof}
Define the realized variance as $Y_{t}\triangleq\int_{0}^{t}\sigma_{s}^{2}ds$
for every $t\in\left[0,T\right]$. Note that $v_{t}^{\delta}=g^{\delta}\left(t,M_{t},Y_{t}\right)$,
where the dynamics of $Y$ and $M$ are given by
\begin{align*}
dY_{t} & =\sigma_{t}^{2}dt,\\
dM_{t} & =\nu A\left(T,t\right)\sqrt{\sigma_{t}^{2}}dW_{t}+\int_{0}^{\infty}c_{3}\eta z\tilde{N}\left(dt,dz\right).
\end{align*}

Now following \cite{DOP094}, we can apply the multidimensional It\^{o}
formula for It\^{o}-L\'evy processes to obtain 
\begin{align*}
dv_{t}^{\delta} & =dg^{\delta}\left(t,M_{t},Y_{t}\right)\\
 & =\frac{\partial g^{\delta}}{\partial t}\left(t,M_{t},Y_{t}\right)dt+\frac{\partial g^{\delta}}{\partial m}\left(t,M_{t},Y_{t}\right)dM_{t}^{c}+\frac{\partial g^{\delta}}{\partial y}\left(t,M_{t},Y_{t}\right)dY_{t}\\
 & \qquad+\frac{1}{2}\frac{\partial^{2}g^{\delta}}{\partial m^{2}}\left(t,M_{t},Y_{t}\right)d\left[M^{c},M^{c}\right]_{t}+\int_{0}^{\infty}\Delta_{m}^{2}g^{\delta}\left(t,M_{t-},Y_{t-}\right)\ell\left(dz\right)dt\\
 & \qquad+\int_{0}^{\infty}\Delta_{m}g^{\delta}\left(t,M_{t-},Y_{t-}\right)\tilde{N}\left(dt,dz\right)\\
 & =\frac{(v_{t}^{\delta})^{2}}{2v_{t}^{\delta}(T-t)}dt+\frac{dM_{t}^{c}}{2v_{t}^{\delta}(T-t)}-\frac{\sigma_{t}^{2}}{2v_{t}^{\delta}(T-t)}dt-\frac{d\left[M^{c},M^{c}\right]_{t}}{8\left(v_{t}^{\delta}\right)^{3}\left(T-t\right)^{2}}\\
 & \qquad+\int_{0}^{\infty}\Delta_{m}^{2}g^{\delta}\left(t,M_{t-},Y_{t-}\right)\ell\left(dz\right)dt+\int_{0}^{\infty}\Delta_{m}g^{\delta}\left(t,M_{t-},Y_{t-}\right)\tilde{N}\left(dt,dz\right)\\
 & =\frac{(v_{t}^{\delta})^{2}-\sigma_{t}^{2}}{2v_{t}^{\delta}(T-t)}dt+\frac{dM_{t}^{c}}{2v_{t}^{\delta}(T-t)}-\frac{d\left[M^{c},M^{c}\right]_{t}}{8\left(v_{t}^{\delta}\right)^{3}\left(T-t\right)^{2}}\\
 & \qquad+\int_{0}^{\infty}\Delta_{m}^{2}g^{\delta}\left(t,M_{t-},Y_{t-}\right)\ell\left(dz\right)dt+\int_{0}^{\infty}\Delta_{m}g^{\delta}\left(t,M_{t-},Y_{t-}\right)\tilde{N}\left(dt,dz\right),
\end{align*}
where $\gamma_{m}\left(t,z\right)\triangleq c_{3}\eta z$ in the expression
for $\Delta_{m}^{2}g^{\delta}\left(t,M_{t-},Y_{t-}\right)$.
\end{proof}
\begin{rem}
It will be useful in further results to write down the dynamics for
the continuous parts of $X_{t}$ and $v_{t}^{\delta}$, respectively
as
\begin{align}
dX_{t}^{c} & =\left(r-\zeta\left(\rho_{2},\eta\right)-\frac{1}{2}\sigma_{t}^{2}\right)dt+\sigma_{t}\left(\rho_{1}dW_{t}+\sqrt{1-\rho_{1}^{2}}d\tilde{W}_{t}\right),\label{eq: log-price_SDE(cont.version)}
\end{align}
\begin{align}
d\left(v_{t}^{\delta}\right)^{c} & =\left[\frac{(v_{t}^{\delta})^{2}-\sigma_{t}^{2}}{2v_{t}^{\delta}(T-t)}+\int_{0}^{\infty}\Delta_{m}^{2}g^{\delta}\left(t,M_{t-},Y_{t-}\right)\ell\left(dz\right)\right]dt\label{eq: fut_avg_vol_contversion_dynamics}\\
 & \qquad+\frac{dM_{t}^{c}}{2v_{t}^{\delta}(T-t)}-\frac{d\left[M^{c},M^{c}\right]_{t}}{8\left(v_{t}^{\delta}\right)^{3}\left(T-t\right)^{2}}.\nonumber 
\end{align}
\end{rem}

\section{\label{sec: General-expansion-formulas}General expansion formulas}

In this section we present the main result of the paper. We provide
an exact expansion formula for option pricing in terms of the Black-Scholes
formula adjusted by extra terms depending on the future expected variance,
the L\'evy measure of the jumps and correlation parameters.
\begin{thm}
\label{thm: Generic_Decomposition_Formula}Let $B=\left\{ B_{t},t\in\left[0,T\right]\right\} $
be a continuous semimartingale with respect to the filtration $\mathcal{F}^{W}\vee\mathcal{F}^{J}$,
let $A\left(t,x,y\right)$ be a $C^{1,2,2}\left(\left[0,T\right]\times\mathbb{R}\times\mathbb{R}\right)$
function such that 
\[
\mathcal{L}_{y}A=\left(\partial_{t}+\frac{1}{2}y^{2}\partial_{xx}^{2}+\left(r-\frac{1}{2}y^{2}\right)\partial_{x}-r\right)A=0,
\]
 and let $v_{t}$ and $M_{t}$ be defined as in the previous section.
Then, for every $t\in\left[0,T\right]$, the expectation of $e^{-rT}A\left(T,X_{T},v_{T}\right)B_{T}$
can be written as follows:
\begin{align*}
 & e^{-r\left(T-t\right)}\mathbb{E}_{t}\left[A\left(T,X_{T},v_{T}\right)B_{T}\right]\\
 & =A\left(t,X_{t},v_{t}\right)B_{t}-\zeta\left(\rho_{2},\eta\right)\mathbb{E}_{t}\left[\int_{t}^{T}e^{-rs}\partial_{x}A\left(s,X_{s},v_{s}\right)B_{s}ds\right]\\
 & \quad+\frac{1}{2}\mathbb{E}_{t}\left[\int_{t}^{T}e^{-r\left(s-t\right)}\left(\partial_{x}^{2}-\partial_{x}\right)A\left(s,X_{s},v_{s}\right)B_{s}\left(\sigma_{s}^{2}-v_{s}^{2}\right)ds\right]\\
 & \quad+\frac{1}{2}\mathbb{E}_{t}\left[\int_{t}^{T}e^{-r\left(s-t\right)}\partial_{y}A\left(s,X_{s},v_{s}\right)B_{s}\left[\frac{v_{s}^{2}-\sigma_{s}^{2}}{v_{s}^{\delta}(T-s)}+2\int_{0}^{\infty}\Delta_{m}^{2}g^{\delta}\left(s,M_{s-},Y_{s-}\right)\ell\left(dz\right)\right]ds\right]\\
 & \quad-\frac{1}{8}\mathbb{E}_{t}\left[\int_{t}^{T}e^{-r\left(s-t\right)}\partial_{y}A\left(s,X_{s},v_{s}\right)B_{s}\left[\frac{d\left[M^{c},M^{c}\right]_{s}}{v_{s}^{3}\left(T-s\right)^{2}}\right]\right]\\
 & \quad+\mathbb{E}_{t}\left[\int_{t}^{T}e^{-r\left(s-t\right)}A\left(s,X_{s},v_{s}\right)dB_{s}\right]\\
 & \quad+\frac{\rho_{1}}{2}\mathbb{E}_{t}\left[\int_{t}^{T}e^{-r\left(s-t\right)}\partial_{xy}^{2}A\left(s,X_{s},v_{s}\right)B_{s}\frac{\sigma_{s}}{v_{s}\left(T-s\right)}d\left[W,M^{c}\right]_{s}\right]\\
 & \quad+\frac{1}{8}\mathbb{E}_{t}\left[\int_{t}^{T}e^{-r\left(s-t\right)}\partial_{y}^{2}A\left(s,X_{s},v_{s}\right)B_{s}\frac{d\left[M^{c},M^{c}\right]_{s}}{v_{s}^{2}\left(T-s\right)^{2}}\right]\\
 & \quad+\rho_{1}\mathbb{E}_{t}\left[\int_{t}^{T}e^{-r\left(s-t\right)}\partial_{x}A\left(s,X_{s},v_{s}\right)\sigma_{s}d\left[W,B\right]_{s}\right]\\
 & \quad+\sqrt{1-\rho_{1}^{2}}\mathbb{E}_{t}\left[\int_{t}^{T}e^{-r\left(s-t\right)}\partial_{x}A\left(s,X_{s},v_{s}\right)\sigma_{s}d\left[\tilde{W},B\right]_{s}\right]\\
 & \quad+\mathbb{E}_{t}\left[\int_{t}^{T}e^{-r\left(s-t\right)}\partial_{y}A\left(s,X_{s},v_{s}\right)\frac{d\left[M^{c},B\right]_{s}}{2v_{s}\left(T-s\right)}\right]\\
 & \quad+\mathbb{E}_{t}\left[\int_{t}^{T}\int_{0}^{\infty}e^{-r\left(s-t\right)}B_{s}\left[\Delta_{x}^{2}A\left(s,X_{s-},v_{s-}\right)+\Delta_{y}^{2}A\left(s,X_{s-},v_{s-}\right)\right]\ell(dz)ds\right].
\end{align*}
\end{thm}

\begin{proof}
Let $F\left(t,X_{t},v_{t}^{\delta}\right)\triangleq e^{-rt}A\left(t,X_{t},v_{t}^{\delta}\right)B_{t}$,
and apply again, the multidimensional It\^{o} formula for L\'evy processes
to $F\left(t,X_{t},v_{t}^{\delta}\right).$ To do so, we will consider
the continuous parts of $X_{t}$ and $v_{t}^{\delta}$, respectively
given by equations $\left(\ref{eq: log-price_SDE(cont.version)}\right)$
and $\left(\ref{eq: fut_avg_vol_contversion_dynamics}\right)$. Therefore
we can write the It\^{o} formula in its integral version over the time
interval $\left[t,T\right]$ as follows:
\begin{align*}
 & e^{-rT}A\left(T,X_{T},v_{T}^{\delta}\right)B_{T}\\
 & =e^{-rt}A\left(t,X_{t},v_{t}^{\delta}\right)B_{t}-r\int_{t}^{T}e^{-rs}A\left(s,X_{s},v_{s}^{\delta}\right)B_{s}ds\\
 & \qquad+\int_{t}^{T}e^{-rs}\partial_{s}A\left(s,X_{s},v_{s}^{\delta}\right)B_{s}ds+\int_{t}^{T}e^{-rs}\partial_{x}A\left(s,X_{s},v_{s}^{\delta}\right)B_{s}dX_{s}^{c}\\
 & \qquad+\int_{t}^{T}e^{-rs}\partial_{y}A\left(s,X_{s},v_{s}^{\delta}\right)B_{s}d\left(v_{s}^{\delta}\right)^{c}+\int_{t}^{T}e^{-rs}A\left(s,X_{s},v_{s}^{\delta}\right)dB_{s}\\
 & \qquad+\frac{1}{2}\int_{t}^{T}e^{-rs}\partial_{x}^{2}A\left(s,X_{s},v_{s}^{\delta}\right)B_{s}d\left[X^{c},X^{c}\right]_{s}\\
 & \qquad+\int_{t}^{T}e^{-rs}\partial_{xy}^{2}A\left(s,X_{s},v_{s}^{\delta}\right)B_{s}d\left[X^{c},\left(v^{\delta}\right)^{c}\right]_{s}\\
 & \qquad+\frac{1}{2}\int_{t}^{T}e^{-rs}\partial_{y}^{2}A\left(s,X_{s},v_{s}^{\delta}\right)B_{s}d\left[\left(v^{\delta}\right)^{c},\left(v^{\delta}\right)^{c}\right]_{s}\\
 & \qquad+\int_{t}^{T}e^{-rs}\partial_{x}A\left(s,X_{s},v_{s}^{\delta}\right)d\left[X^{c},B\right]_{s}\\
 & \qquad+\int_{t}^{T}e^{-rs}\partial_{y}A\left(s,X_{s},v_{s}^{\delta}\right)d\left[\left(v^{\delta}\right)^{c},B\right]_{s}\\
 & \qquad+\int_{t}^{T}\int_{0}^{\infty}e^{-rs}B_{s}\left[\Delta_{x}^{2}A\left(s,X_{s-},v_{s-}^{\delta}\right)+\Delta_{y}^{2}A\left(s,X_{s-},v_{s-}^{\delta}\right)\right]\ell(dz)ds\\
 & \qquad+\int_{t}^{T}\int_{0}^{\infty}e^{-rs}B_{s}\left[\Delta_{x}A\left(s,X_{s-},v_{s-}^{\delta}\right)+\Delta_{y}A\left(s,X_{s-},v_{s-}^{\delta}\right)\right]\tilde{N}\left(ds,dz\right).
\end{align*}
Now, recalling the definitions of $dX_{t}^{c}$ and $d\left(v_{t}^{\delta}\right)^{c}$,
the fact that
\begin{align*}
d\left[X^{c},X^{c}\right]_{s} & =\sigma_{s}^{2}ds,\\
d\left[X^{c},\left(v^{\delta}\right)^{c}\right]_{s} & =\frac{\sigma_{s}}{2v_{s}^{\delta}\left(T-s\right)}\left(\rho_{1}d\left[W,M^{c}\right]_{s}+\sqrt{1-\rho_{1}^{2}}d\left[\tilde{W},M^{c}\right]_{s}\right),\\
d\left[\left(v^{\delta}\right)^{c},\left(v^{\delta}\right)^{c}\right]_{s} & =\frac{d\left[M^{c},M^{c}\right]_{s}}{4\left(v_{s}^{\delta}\right)^{2}\left(T-s\right)^{2}},\\
d\left[X^{c},B\right]_{s} & =\sigma_{s}\left(\rho_{1}d\left[W,B\right]_{s}+\sqrt{1-\rho_{1}^{2}}d\left[\tilde{W},B\right]_{s}\right),\\
d\left[\left(v^{\delta}\right)^{c},B\right]_{s} & =\frac{d\left[M^{c},B\right]_{s}}{2v_{s}^{\delta}\left(T-s\right)},
\end{align*}
and the independence between $M$ and $\tilde{W}$, we can rewrite
$e^{-rT}A\left(T,X_{T},v_{T}^{\delta}\right)B_{T}$ as
\begin{align*}
 & e^{-rT}A\left(T,X_{T},v_{T}^{\delta}\right)B_{T}\\
 & \qquad=e^{-rt}A\left(t,X_{t},v_{t}^{\delta}\right)B_{t}\\
 & \qquad+\int_{t}^{T}e^{-rs}\left(\partial_{s}+\frac{1}{2}\sigma_{s}^{2}\partial_{x}^{2}+\left(r-\frac{1}{2}\sigma_{s}^{2}\right)\partial_{x}-r\right)A\left(s,X_{s},v_{s}^{\delta}\right)B_{s}ds\\
 & \qquad-\zeta\left(\rho_{2},\eta\right)\int_{t}^{T}e^{-rs}\partial_{x}A\left(s,X_{s},v_{s}^{\delta}\right)B_{s}ds\\
 & \qquad+\int_{t}^{T}e^{-rs}\partial_{x}A\left(s,X_{s},v_{s}^{\delta}\right)B_{s}\left[\sigma_{s}\left(\rho_{1}dW_{s}+\sqrt{1-\rho_{1}^{2}}d\tilde{W}_{s}\right)\right]\\
 & \qquad+\frac{1}{2}\int_{t}^{T}e^{-rs}\partial_{y}A\left(s,X_{s},v_{s}^{\delta}\right)B_{s}\left[\frac{(v_{s}^{\delta})^{2}-\sigma_{s}^{2}}{v_{s}^{\delta}(T-s)}+2\int_{0}^{\infty}\Delta_{m}^{2}g^{\delta}\left(s,M_{s-},Y_{s-}\right)\ell\left(dz\right)\right]ds\\
 & \qquad+\frac{1}{2}\int_{t}^{T}e^{-rs}\partial_{y}A\left(s,X_{s},v_{s}^{\delta}\right)B_{s}\left[\frac{dM_{s}^{c}}{v_{s}^{\delta}(T-s)}\right]\\
 & \qquad-\frac{1}{8}\int_{t}^{T}e^{-rs}\partial_{y}A\left(s,X_{s},v_{s}^{\delta}\right)B_{s}\left[\frac{d\left[M^{c},M^{c}\right]_{s}}{\left(v_{s}^{\delta}\right)^{3}\left(T-s\right)^{2}}\right]\\
 & \qquad+\int_{t}^{T}e^{-rs}A\left(s,X_{s},v_{s}^{\delta}\right)dB_{s}\\
 & \qquad+\frac{\rho_{1}}{2}\int_{t}^{T}e^{-rs}\partial_{xy}^{2}A\left(s,X_{s},v_{s}^{\delta}\right)B_{s}\frac{\sigma_{s}}{v_{s}^{\delta}\left(T-s\right)}d\left[W,M^{c}\right]_{s}\\
 & \qquad+\frac{1}{2}\int_{t}^{T}e^{-rs}\partial_{y}^{2}A\left(s,X_{s},v_{s}^{\delta}\right)B_{s}\frac{d\left[M^{c},M^{c}\right]_{s}}{4\left(v_{s}^{\delta}\right)^{2}\left(T-s\right)^{2}}\\
 & \qquad+\rho_{1}\int_{t}^{T}e^{-rs}\partial_{x}A\left(s,X_{s},v_{s}^{\delta}\right)\sigma_{s}d\left[W,B\right]_{s}\\
 & \qquad+\sqrt{1-\rho_{1}^{2}}\int_{t}^{T}e^{-rs}\partial_{x}A\left(s,X_{s},v_{s}^{\delta}\right)\sigma_{s}d\left[\tilde{W},B\right]_{s}\\
 & \qquad+\int_{t}^{T}e^{-rs}\partial_{y}A\left(s,X_{s},v_{s}^{\delta}\right)\frac{d\left[M^{c},B\right]_{s}}{2v_{s}^{\delta}\left(T-s\right)}\\
 & \qquad+\int_{t}^{T}\int_{0}^{\infty}e^{-rs}B_{s}\left[\Delta_{x}^{2}A\left(s,X_{s-},v_{s-}^{\delta}\right)+\Delta_{y}^{2}A\left(s,X_{s-},v_{s-}^{\delta}\right)\right]\ell(dz)ds\\
 & \qquad+\int_{t}^{T}\int_{0}^{\infty}e^{-rs}B_{s}\left[\Delta_{x}A\left(s,X_{s-},v_{s-}^{\delta}\right)+\Delta_{y}A\left(s,X_{s-},v_{s-}^{\delta}\right)\right]\tilde{N}\left(ds,dz\right).
\end{align*}
We can identify in the previous expression, the operator $\mathcal{L}_{\sigma_{s}}$.
We know from \cite{JafaVives13}, that the following is true, 
\[
\mathcal{L}_{\sigma_{s}}=\mathcal{L}_{v_{s}^{\delta}}+\frac{1}{2}\left(\sigma_{s}^{2}-\left(v_{s}^{\delta}\right)^{2}\right)\left(\partial_{x}^{2}-\partial_{x}\right),
\]
where $\mathcal{L}_{v_{s}^{\delta}}A=0$. Therefore, multiplying by
$e^{rt}$, taking conditional expectation and letting $\delta\searrow0$,
combined with the use of the dominated convergence theorem, we obtain
result follows, ending the proof.
\end{proof}
If we assume additional properties on the function $A$ one can simplify
the formula given in the previous theorem. In the following corollary,
we assume certain relationship of the partial derivative with respect
to $y$ and the first and second order partial derivatives with respect
to $x$. This relation is often referred in the literature as the
Delta-Gamma-Vega relationship and it is satisfied by the Black-Scholes
function.
\begin{cor}
\label{cor: Funct_DGV_Relationship}Let the function $A$ and the
process $B$, be defined as in Theorem \ref{thm: Generic_Decomposition_Formula}.
Suppose that the function $A$ satisfies the Delta-Gamma-Vega relationship
given by
\begin{equation}
\partial_{y}A\left(t,x,y\right)\frac{1}{y\left(T-t\right)}=\left(\partial_{xx}^{2}-\partial_{x}\right)A\left(t,x,y\right).\label{eq: Functional_DGV_Relationship}
\end{equation}
Then, for every $t\in\left[0,T\right]$, the following formula holds:
\begin{align*}
 & e^{-r\left(T-t\right)}\mathbb{E}_{t}\left[A\left(T,X_{T},v_{T}\right)B_{T}\right]\\
 & \qquad=A\left(t,X_{t},v_{t}\right)B_{t}-\zeta\left(\rho_{2},\eta\right)\mathbb{E}_{t}\left[\int_{t}^{T}e^{-rs}\Lambda A\left(s,X_{s},v_{s}\right)B_{s}ds\right]\\
 & \qquad+\mathbb{E}_{t}\left[\int_{t}^{T}e^{-r\left(s-t\right)}\Gamma A\left(s,X_{s},v_{s}\right)B_{s}\left[v_{s}\left(T-s\right)\int_{0}^{\infty}\Delta_{m}^{2}g\left(s,M_{s-},Y_{s-}\right)\ell\left(dz\right)\right]ds\right]\\
 & \qquad+\mathbb{E}_{t}\left[\int_{t}^{T}e^{-r\left(s-t\right)}A\left(s,X_{s},v_{s}\right)dB_{s}\right]\\
 & \qquad+\frac{\rho_{1}}{2}\mathbb{E}_{t}\left[\int_{t}^{T}e^{-r\left(s-t\right)}\Lambda\Gamma A\left(s,X_{s},v_{s}\right)B_{s}\sigma_{s}d\left[W,M^{c}\right]_{s}\right]\\
 & \qquad+\frac{1}{8}\mathbb{E}_{t}\left[\int_{t}^{T}e^{-r\left(s-t\right)}\Gamma^{2}A\left(s,X_{s},v_{s}\right)B_{s}d\left[M^{c},M^{c}\right]_{s}\right]\\
 & \qquad+\rho_{1}\mathbb{E}_{t}\left[\int_{t}^{T}e^{-r\left(s-t\right)}\Lambda A\left(s,X_{s},v_{s}\right)\sigma_{s}d\left[W,B\right]_{s}\right]\\
 & \qquad+\sqrt{1-\rho_{1}^{2}}\mathbb{E}_{t}\left[\int_{t}^{T}e^{-r\left(s-t\right)}\Lambda A\left(s,X_{s},v_{s}\right)\sigma_{s}d\left[\tilde{W},B\right]_{s}\right]\\
 & \qquad+\frac{1}{2}\mathbb{E}_{t}\left[\int_{t}^{T}e^{-r\left(s-t\right)}\Gamma A\left(s,X_{s},v_{s}\right)d\left[M^{c},B\right]_{s}\right]\\
 & \qquad+\mathbb{E}_{t}\left[\int_{t}^{T}\int_{0}^{\infty}e^{-r\left(s-t\right)}B_{s}\left[\Delta_{x}^{2}A\left(s,X_{s-},v_{s-}\right)+\Delta_{y}^{2}A\left(s,X_{s-},v_{s-}\right)\right]\ell(dz)ds\right].
\end{align*}
\end{cor}

\begin{proof}
If $\left(\ref{eq: Functional_DGV_Relationship}\right)$ holds, then
it is trivial to see that
\begin{align*}
\partial_{xy}^{2}A\left(t,x,y\right) & =y\left(T-t\right)\left(\partial_{xx}^{3}-\partial_{xx}^{2}\right)A\left(t,x,y\right),\\
\partial_{yy}^{2}A\left(t,x,y\right) & =\frac{y^{2}\left(T-t\right)^{2}}{y^{2}\left(T-t\right)}\left(\partial_{xx}^{2}-\partial_{x}\right)A\left(t,x,y\right)\\
 & \qquad+y^{2}\left(T-t\right)^{2}\left(\partial_{xx}^{2}-\partial_{x}\right)^{2}A\left(t,x,y\right).
\end{align*}

Therefore, by replacing $\left(\ref{eq: Functional_DGV_Relationship}\right)$
together with the previous equalities in Theorem \ref{thm: Generic_Decomposition_Formula},
the result is straightforward.
\end{proof}
The next result yields an analogous formula to Proposition 9 in \cite{AloYan17},
that contains the continuous part of the formula and the discontinuous
terms coming from the jumps assumed in the model.
\begin{thm}
\label{thm: BS expansion formula}Assume the model given by equations
\ref{eq: fHestonJump_Volatility}, such that $2k\theta\geq\nu^{2},$
\[
\left(1-c_{1}-c_{2}\frac{T^{\alpha}}{\alpha\Gamma(\alpha)}\right)\geq0,
\]
 and let $A\left(t,X_{t},v_{t}\right)=BS\left(t,X_{t},v_{t}\right)$
and $B_{t}\equiv1$. Then, for every $t\in\left[0,T\right]$,
\end{thm}

Then we have
\begin{align}
V_{t} & =BS\left(t,X_{t},v_{t}\right)-\zeta\left(\rho_{2},\eta\right)\mathbb{E}_{t}\left[\int_{t}^{T}e^{-rs}\Lambda BS\left(s,X_{s},v_{s}\right)ds\right]\nonumber \\
 & \qquad+\frac{\rho_{1}}{2}\mathbb{E}_{t}\left[\int_{t}^{T}e^{-r\left(s-t\right)}\Lambda\Gamma BS\left(s,X_{s},v_{s}\right)\sigma_{s}d\left[W,M^{c}\right]_{s}\right]\nonumber \\
 & \qquad+\frac{1}{8}\mathbb{E}_{t}\left[\int_{t}^{T}e^{-r\left(s-t\right)}\Gamma^{2}BS\left(s,X_{s},v_{s}\right)d\left[M^{c},M^{c}\right]_{s}\right]\nonumber \\
 & \qquad+\mathbb{E}_{t}\left[\int_{t}^{T}e^{-r\left(s-t\right)}\Gamma BS\left(s,X_{s},v_{s}\right)\left[v_{s}\left(T-s\right)\int_{0}^{\infty}\Delta_{m}^{2}g\left(s,M_{s-},Y_{s-}\right)\ell\left(dz\right)\right]ds\right]\nonumber \\
 & \qquad+\mathbb{E}_{t}\left[\int_{t}^{T}\int_{0}^{\infty}e^{-r\left(s-t\right)}\left[\Delta_{x}^{2}BS\left(s,X_{s-},v_{s-}\right)+\Delta_{y}^{2}BS\left(s,X_{s-},v_{s-}\right)\right]\ell(dz)ds\right]\nonumber \\
 & \triangleq BS\left(t,X_{t},v_{t}\right)-\zeta\left(\rho_{2},\eta\right)\int_{t}^{T}e^{-rs}\Lambda BS\left(s,X_{s},v_{s}\right)ds+\left(I\right)+\left(II\right)+\left(III\right)+\left(IV\right).\label{eq: BS expansion formula}
\end{align}

\begin{proof}
The result is trivially achieved by replacing $A\left(t,X_{t},v_{t}\right)=BS\left(t,X_{t},v_{t}\right)$
and $B_{t}\equiv1$ in Corollary \ref{cor: Funct_DGV_Relationship},
and noticing that $V_{t}=e^{-r\left(T-t\right)}\mathbb{E}_{t}\left[BS\left(T,X_{T},v_{T}^{\delta}\right)\right]$.
\end{proof}

\section{\label{sec: Martingale_Representation}Martingale Representation
of the future expected volatility}

This section is devoted to derive an expression for the dynamics of
the integrated projected future variance $M$. In the next proposition
we show that, under certain conditions on the parameters, the variance
process is bounded away from zero (lower bounded by a strictly positive
function).
\begin{prop}
Consider $\alpha\in\left(0,\frac{1}{2}\right)$ and $T\geq0.$ Assuming
that $2k\theta\geq\nu^{2}$ and 
\[
\left(1-c_{1}-c_{2}\frac{T^{\alpha}}{\alpha\Gamma(\alpha)}\right)\geq0.
\]
 Then for all $0<t<T$
\begin{equation}
\sigma_{t}^{2}\geq\bar{\sigma}_{0}^{2}e^{-\kappa t}+\theta(1-e^{-\kappa t})\left(1-c_{1}-c_{2}\frac{T^{\alpha}}{\alpha\Gamma(\alpha)}\right)\,\,a.s.\label{eq: Vol_LowerBound}
\end{equation}
\end{prop}

\begin{proof}
We know by definition of the volatility equation given by $\left(\ref{eq: fHestonJump_Volatility}\right)$,
that
\[
\sigma_{t}^{2}=U_{t}+c_{1}\nu Z_{t}+c_{2}\nu\frac{1}{\Gamma(\alpha)}\int_{0}^{t}(t-r)^{\alpha-1}Z_{r}dr+c_{3}\eta J_{t}.
\]
Knowing that jumps are only positive and that $c_{3}\geq0$, $\eta>0$,
we can lower bound the volatility process by
\[
\sigma_{t}^{2}\geq U_{t}+c_{1}\nu Z_{t}+c_{2}\nu\frac{1}{\Gamma(\alpha)}\int_{0}^{t}(t-r)^{\alpha-1}Z_{r}dr.
\]
Back to the Heston's volatility process, we know that it has to be
positive, therefore $\nu Z_{t}>-U_{t}=-(\theta+e^{-\kappa t}(\bar{\sigma}_{0}^{2}-\theta))$
for all initial condition $\bar{\sigma}_{0}^{2}.$ And letting $\bar{\sigma}_{0}^{2}\rightarrow0$
we have that $\nu Z_{t}\geq-\theta(1-e^{-\kappa t})$ a.s. Now we
can write
\begin{align*}
\sigma_{t}^{2} & \geq U_{t}-c_{1}\theta(1-e^{-\kappa t})-c_{2}\frac{\theta}{\Gamma(\alpha)}\int_{0}^{t}(t-r)^{\alpha-1}(1-e^{-\kappa r})dr\\
 & \geq U_{t}-c_{1}\theta(1-e^{-\kappa t})-c_{2}\frac{\theta(1-e^{-\kappa t})}{\Gamma(\alpha)}\int_{0}^{t}(t-r)^{\alpha-1}dr\\
 & =U_{t}-c_{1}\theta(1-e^{-\kappa t})-c_{2}\frac{\theta(1-e^{-\kappa t})t^{\alpha}}{\alpha\Gamma(\alpha)}\\
 & =\bar{\sigma}_{0}^{2}e^{-\kappa t}+\theta(1-e^{-\kappa t})\left(1-c_{1}-c_{2}\frac{t^{\alpha}}{\alpha\Gamma(\alpha)}\right)\\
 & \geq\bar{\sigma}_{0}^{2}e^{-\kappa t}+\theta(1-e^{-\kappa t})\left(1-c_{1}-c_{2}\frac{T^{\alpha}}{\alpha\Gamma(\alpha)}\right).
\end{align*}
And it is a positive quantity since we have by hypothesis that
\[
\left(1-c_{1}-c_{2}\frac{T^{\alpha}}{\alpha\Gamma(\alpha)}\right)\geq0.
\]
\end{proof}
Now as we have defined $v_{t}^{2}=\frac{1}{T-t}\left(M_{t}-\int_{0}^{t}\sigma_{s}^{2}ds\right)$,
when applying the It\^{o} Lemma we need to compute $dv_{t}^{2}$ which
implies computing $dM_{t}$, for $M_{t}=\int_{0}^{T}\mathbb{E}_{t}\left[\sigma_{s}^{2}\right]ds$.
The computation of this last derivative, needs to be done by means
of Clark-Ocone-Haussmann formula and Malliavin calculus techniques.
A good reference for this topic is, for instance, \cite{DOP094}.
\begin{prop}
\label{prop: Malliavin_Derivative(JumpCase)}We have that $\sigma_{t}^{2}\in\mathbb{D}_{N}^{1,2}$
and $D_{s,z}^{N}\sigma_{t}^{2}=c_{3}\eta z.$
\end{prop}

\begin{proof}
Note that $\sigma_{t}^{2}=Y_{t}+c_{1}\nu Z_{t}+c_{2}\nu\frac{1}{\Gamma(\alpha)}\int_{0}^{t}(t-r)^{\alpha-1}Z_{r}dr+c_{3}\eta J_{t}.$
Now we have that
\begin{align*}
D_{s,z}^{N}\sigma_{t}^{2} & =D_{s,y}^{N}\left(Y_{t}+c_{1}\nu Z_{t}+c_{2}\nu\frac{1}{\Gamma(\alpha)}\int_{0}^{t}(t-r)^{\alpha-1}Z_{r}dr+c_{3}\eta J_{t}\right)\\
 & =D_{s,z}^{N}\,c_{3}\eta J_{t},
\end{align*}
and since $J_{t}$ is a pure jump L\'evy process which can be represented
as $J_{t}=-\zeta\left(\rho_{2},\eta\right)t+\int_{0}^{t}\int_{0}^{\infty}z\tilde{N}(ds,dz)$,
following \cite{DOP094}, we have that
\begin{align*}
D_{s,z}^{N}\sigma_{t}^{2} & =D_{s,z}^{N}c_{3}\eta J_{t}^{0}=D_{s,z}^{N}c_{3}\eta\left(-\zeta\left(\rho_{2},\eta\right)t+\int_{0}^{t}\int_{0}^{\infty}z\tilde{N}(ds,dz)\right)\\
 & =D_{s,z}^{N}\left(\int_{0}^{t}\int_{0}^{\infty}c_{3}\eta z\tilde{N}(ds,dz)\right)=c_{3}\eta z.
\end{align*}
\end{proof}
\begin{prop}
\label{prop: Malliavin_Derivative(BrownianCase)}Assume that $2k\theta\geq\nu^{2}$.
Then, we have that $\sigma_{t}^{2}\in\mathbb{D}_{W}^{1,2}$ and 
\begin{equation}
D_{s}^{W}\sigma_{t}^{2}=c_{1}D_{s}^{W}\bar{\sigma}_{t}^{2}+\frac{c_{2}}{\Gamma(\alpha)}\int_{s}^{t}(t-r)^{\alpha-1}D_{s}^{W}\bar{\sigma}_{r}^{2}dr.\label{eq: BM_MalliavinDerivative}
\end{equation}
\end{prop}

\begin{proof}
Starting from the definition of $\sigma_{t}^{2}$, we have that 
\begin{align*}
D_{s}^{W}\sigma_{t}^{2} & =D_{s}^{W}\left(Y_{t}+c_{1}\nu Z_{t}+c_{2}\nu I_{0+}^{\alpha}Z_{t}+c_{3}\eta J_{t}\right)\\
 & =c_{1}\nu D_{s}^{W}Z_{t}+c_{2}I_{0+}^{\alpha}\nu D_{s}^{W}Z_{t}.
\end{align*}

We have to note that $D_{s}^{W}J_{t}=0$ and that $D_{s}^{W}\bar{\sigma}_{t}^{2}=\nu D_{s}^{W}Z_{t}$,
hence we can rewrite the previous expression as 
\[
D_{s}^{W}\sigma_{t}^{2}=c_{1}D_{s}^{W}\bar{\sigma}_{t}^{2}+\frac{c_{2}}{\Gamma(\alpha)}\int_{s}^{t}(t-r)^{\alpha-1}D_{s}^{W}\bar{\sigma}_{r}^{2}dr.
\]
Now, for instance, we have to solve $D_{s}^{W}\bar{\sigma}_{t}^{2}$.
To achieve this we will make use of Theorem 2.1 in \cite{Detemple2005}
where in our particular case $\mu(s,\bar{\sigma}_{t}^{2})\triangleq\kappa(\theta-\bar{\sigma}_{t}^{2})$
and $\sigma(s,\bar{\sigma}_{t}^{2})\triangleq\nu\sqrt{\bar{\sigma}_{t}^{2}}$.
Therefore, we have that by replacing and doing some basic algebraic
manipulations we obtain the following expression.
\begin{align*}
D_{s}^{W}\bar{\sigma}_{t}^{2} & =\sigma\,\exp\left\{ \int_{s}^{t}\left[\partial_{2}\mu-\frac{\mu\partial_{2}\sigma}{\sigma}-\frac{1}{2}(\partial_{22}\sigma)\sigma-\frac{\partial_{1}\sigma}{\sigma}\right](u,\bar{\sigma}_{u}^{2})du\right\} \\
 & =\nu\,\sqrt{\bar{\sigma}_{t}^{2}}\exp\left\{ \int_{s}^{t}\left[-\frac{\kappa}{2}-\left(\frac{\kappa\theta}{2}-\frac{\nu^{2}}{8}\right)\frac{1}{\bar{\sigma}_{u}^{2}}\right]du\right\} \\
 & =\nu\,\sqrt{\bar{\sigma}_{t}^{2}}f(t,s),
\end{align*}
where $f(t,s)\triangleq\exp\left\{ \int_{s}^{t}\left[-\frac{\kappa}{2}-\left(\frac{\kappa\theta}{2}-\frac{\nu^{2}}{8}\right)\frac{1}{\bar{\sigma}_{u}^{2}}\right]du\right\} $.
\end{proof}
\begin{prop}
\label{prop: Martingale_Dynamics}Assume that $2k\theta\geq\nu^{2}$
and $\left(1-c_{1}-c_{2}\frac{T^{\alpha}}{\alpha\Gamma(\alpha)}\right)\geq0$.
Then, 
\begin{equation}
dM_{t}=\nu A(T,t)\sqrt{\bar{\sigma}_{t}^{2}}dW_{t}+c_{3}\eta dJ_{t}.\label{eq: Martingale_Dynamics}
\end{equation}
\end{prop}

\begin{proof}
Using the Clark-Ocone-Haussmann formula we have
\begin{equation}
\sigma_{t}^{2}=\mathbb{E}\left[\sigma_{t}^{2}\right]+\int_{0}^{t}\mathbb{E}_{s}\left[D_{s}^{W}\sigma_{t}^{2}\right]dW_{s}+\int_{0}^{t}\int_{0}^{\infty}\mathbb{E}_{s}\left[D_{s,z}^{N}\sigma_{t}^{2}\right]\tilde{N}(ds,dz).\label{eq:Clark-Ocone-Haussmann}
\end{equation}
Now we will treat the second addend using \ref{eq: BM_MalliavinDerivative}
and doing some algebraic manipulations, to obtain the following relationship
\begin{align*}
\int_{0}^{t}\mathbb{E}_{s}\left[D_{s}^{W}\sigma_{t}^{2}\right]dW_{s} & =\int_{0}^{t}\mathbb{E}_{s}\left[c_{1}D_{s}^{W}\bar{\sigma}_{t}^{2}+\frac{c_{2}}{\Gamma(\alpha)}\int_{s}^{t}(t-r)^{\alpha-1}D_{s}^{W}\bar{\sigma}_{r}^{2}dr\right]dW_{s}\\
 & =\int_{0}^{t}\left[c_{1}\mathbb{E}_{s}\left[D_{s}^{W}\bar{\sigma}_{t}^{2}\right]+\frac{c_{2}}{\Gamma(\alpha)}\int_{s}^{t}(t-r)^{\alpha-1}\mathbb{E}_{s}\left[D_{s}^{W}\bar{\sigma}_{r}^{2}\right]dr\right]dW_{s}.
\end{align*}
Remembering from proposition $\left(\ref{prop: Malliavin_Derivative(BrownianCase)}\right)$
that $D_{s}^{W}\bar{\sigma}_{t}^{2}=\nu D_{s}^{W}Z_{t}$ we can rewrite
the previous equation as
\begin{align*}
 & \int_{0}^{t}\mathbb{E}_{s}\left[D_{s}^{W}\sigma_{t}^{2}\right]dW_{s}\\
 & =\int_{0}^{t}\left[c_{1}\nu\mathbb{E}_{s}\left[D_{s}^{W}Z_{t}\right]+\frac{c_{2}\nu}{\Gamma(\alpha)}\int_{s}^{t}(t-r)^{\alpha-1}\mathbb{E}_{s}\left[D_{s}^{W}Z_{r}\right]dr\right]dW_{s}\\
 & =\int_{0}^{t}\left[c_{1}\nu\,\exp\left\{ -\kappa(t-s)\right\} +\left(\frac{c_{2}\nu}{\Gamma(\alpha)}\int_{s}^{t}(t-r)^{\alpha-1}\exp\left\{ -\kappa(t-s)\right\} dr\right)\right]\sqrt{\bar{\sigma}_{s}^{2}}dW_{s},
\end{align*}
as $Z_{t}=\int_{0}^{t}e^{-\kappa(t-s)}\sqrt{\bar{\sigma}_{s}^{2}}dW_{s}$.
Finally if we set
\[
a(t,s):=\left[c_{1}\nu\,\exp\left\{ -\kappa(t-s)\right\} +\left(\frac{c_{2}\nu}{\Gamma(\alpha)}\int_{s}^{t}(t-r)^{\alpha-1}\exp\left\{ -\kappa(t-s)\right\} dr\right)\right]\sqrt{\bar{\sigma}_{s}^{2}},
\]
we can write $\int_{0}^{t}\mathbb{E}_{s}\left[D_{s}^{W}\sigma_{t}^{2}\right]dW_{s}$
in a more compact way as
\[
\int_{0}^{t}\mathbb{E}_{s}\left[D_{s}^{W}\sigma_{t}^{2}\right]dW_{s}=\int_{0}^{t}a(t,s)dW_{s}.
\]

We will continue by treating the third addend in equation \ref{eq:Clark-Ocone-Haussmann},
using proposition $\left(\ref{prop: Malliavin_Derivative(JumpCase)}\right)$.

\[
\int_{0}^{t}\int_{0}^{\infty}\mathbb{E}_{s}\left[D_{s,z}^{N}\sigma_{t}^{2}\right]\tilde{N}(ds,dz)=\int_{0}^{t}\int_{0}^{\infty}c_{3}\eta z\tilde{N}(ds,dz),
\]
which is a martingale process. Therefore, recalling that we defined
$M_{t}=\int_{0}^{T}\mathbb{E}_{t}\left[\sigma_{s}^{2}\right]ds$,
we can write $dM_{t}$ as
\begin{align*}
dM_{t} & =\left(\int_{t}^{T}a(r,t)dr\right)dW_{t}+c_{3}\eta d\tilde{J}_{t}\\
 & =\nu\int_{t}^{T}c_{1}\exp\left\{ -\kappa(r-t)\right\} \sqrt{\bar{\sigma}_{r}^{2}}dr\,dW_{t}\\
 & \qquad+\frac{c_{2}\nu}{\Gamma(\alpha)}\int_{t}^{T}\int_{t}^{r}(r-u)^{\alpha-1}\exp\left\{ -\kappa(u-t)\right\} \sqrt{\bar{\sigma}_{r}^{2}}dudr\,dW_{t}+c_{3}\eta d\tilde{J}_{t}.
\end{align*}
Using Fubini's Theorem and integrating the expression we obtain
\begin{align*}
dM_{t} & =\nu\int_{t}^{T}c_{1}\exp\left\{ -\kappa(u-t)\right\} \sqrt{\bar{\sigma}_{t}^{2}}dr\\
 & \qquad+\frac{c_{2}\nu}{\Gamma(\alpha)}\int_{t}^{r}\int_{t}^{T}(r-u)^{\alpha-1}\exp\left\{ -\kappa(u-t)\right\} \sqrt{\bar{\sigma}_{t}^{2}}drdu\,dW_{t}+c_{3}\eta d\tilde{J}_{t}\\
 & =\nu\int_{t}^{T}c_{1}\exp\left\{ -\kappa(u-t)\right\} \sqrt{\bar{\sigma}_{t}^{2}}dr\,dW_{t}\\
 & \qquad+\frac{c_{2}\nu}{\alpha\Gamma(\alpha)}\int_{t}^{r}(T-u)^{\alpha}\exp\left\{ -\kappa(u-t)\right\} \sqrt{\bar{\sigma}_{t}^{2}}du\,dW_{t}+c_{3}\eta d\tilde{J}_{t}\\
 & =\nu\left(\int_{t}^{T}\left(\frac{c_{2}}{\alpha\Gamma(\alpha)}\left(T-u\right)^{\alpha}+c_{1}\right)\exp\left\{ -\kappa(u-t)\right\} du\right)\sqrt{\bar{\sigma}_{t}^{2}}dW_{t}+c_{3}\eta d\tilde{J}_{t}.
\end{align*}
Defining $A(T,t)\triangleq\int_{t}^{T}\left(\frac{c_{2}}{\alpha\Gamma(\alpha)}\left(T-u\right)^{\alpha}+c_{1}\right)\exp\left\{ -\kappa(u-t)\right\} du,$
we can rewrite the previous expression as
\[
dM_{t}=\nu A(T,t)\sqrt{\bar{\sigma}_{t}^{2}}dW_{t}+c_{3}\eta d\tilde{J}_{t}.
\]
\end{proof}

\section{\label{sec: Call-Price-Approximations}Call Price Approximations
in the Fractional Heston Model with Jumps}

This section is aimed at providing an approximation formula for the
pricing of vanilla options presented in Theorem \ref{thm: BS expansion formula},
where we assumed the FSVJJ model presented in previous sections. In
particular, we will deduce a first order approximation formula with
respect to the vol-of-vol parameter $\nu$, and the jump parameter
$\eta$. In order to do so, we will need to introduce a series of
technical lemmas. These will help us bound the conditional expectation
of the integrated future variance, needed to bound the error terms
of the approximating formula. 
\begin{lem}
\label{lem: LD_derivatives}The following results regarding some of
the expressions considered in previous sections hold.
\end{lem}

\begin{enumerate}
\item $A\left(T,t\right)=\int_{t}^{T}\left(\frac{c_{2}}{\Gamma\left(\alpha\right)}\left(T-u\right)^{\alpha}+c_{1}\right)e^{-\kappa\left(u-t\right)}du$.
\item $\mathbb{E}_{t}\left[\bar{\sigma}_{s}^{2}\right]=\theta+\left(\bar{\sigma}_{t}^{2}-\theta\right)e^{-\kappa\left(s-t\right)}=\bar{\sigma}_{t}^{2}e^{-\kappa\left(s-t\right)}+\theta\left(1-e^{-\kappa\left(s-t\right)}\right)$
\item $\mathbb{E}_{t}\left[\sigma_{s}\sqrt{\bar{\sigma}_{s}^{2}}\right]=\bar{\sigma}_{t}^{2}e^{-\kappa\left(s-t\right)}+\theta\left(1-e^{-\kappa\left(s-t\right)}\right)+\mathcal{O}\left(\nu^{2}+\eta^{2}\right).$ 
\item $dM_{t}^{c}=\nu A(T,t)\sqrt{\bar{\sigma}_{t}^{2}}dW_{t}$.
\item $d\left[W,M^{c}\right]_{t}=\nu A(T,t)\sqrt{\bar{\sigma}_{t}^{2}}dt$.
\item $d\left[M^{c},M^{c}\right]_{t}=\nu^{2}A^{2}(T,t)\bar{\sigma}_{t}^{2}dt$.
\item $L\left[W,M^{c}\right]_{t}=\nu\int_{t}^{T}A(T,s)\mathbb{E}_{t}\left[\sqrt{\sigma_{s}^{2}\bar{\sigma}_{s}^{2}}\right]ds$.
\item $D\left[M^{c},M^{c}\right]_{t}=\nu^{2}\int_{t}^{T}A^{2}(T,s)\mathbb{E}_{t}\left[\bar{\sigma}_{s}^{2}\right]ds.$
\item 
\begin{align*}
dL\left[W,M^{c}\right]_{t} & =d\left(\mathbb{E}_{t}\left[\int_{t}^{T}\sigma_{s}d\left[W,M^{c}\right]_{s}\right]\right)\\
 & =\nu^{2}\left(\int_{t}^{T}A(T,s)e^{-\kappa\left(s-t\right)}ds\right)\bar{\sigma}_{t}dW_{t}-\nu\bar{\sigma}_{t}^{2}A\left(T,t\right)dt+\mathcal{O}\left(\nu^{3}+\nu\eta^{2}\right).
\end{align*}
\item 
\begin{align*}
dD\left[M^{c},M^{c}\right]_{t} & =d\left(\mathbb{E}_{t}\left[\int_{t}^{T}d\left[M^{c},M^{c}\right]_{s}\right]\right)\\
 & =\nu^{3}\left(\int_{t}^{T}A^{2}(T,s)e^{-\kappa\left(s-t\right)}ds\right)\bar{\sigma}_{t}dW_{t}-\nu^{2}\bar{\sigma}_{t}^{2}A^{2}\left(T,t\right)dt.
\end{align*}
\end{enumerate}
\begin{proof}
The proof of this result can be found in the subsection \ref{subsec: Proof Lemma LD_derivatives}
of the Appendix. 
\end{proof}
\begin{lem}
\label{lem: BS_BoundedDerivatives}Let $0\leq t<s\leq T$ and $\mathcal{G}_{t}\triangleq\mathcal{F}_{t}\vee\mathcal{F}_{T}^{W}\vee\mathcal{F}_{T}^{J^{0}}$.
Then for every $n\geq0$, there exists $C=C\left(n,\rho_{1}\right)$
such that
\[
\left|\mathbb{E}\left[\partial_{x}^{n}G\left(s,X_{s}+a,v_{s}\right)\mid\mathcal{G}_{t}\right]\right|\leq C\left(\int_{s}^{T}\mathbb{E}_{s}\left[\sigma_{\theta}\right]d\theta\right)^{-\frac{1}{2}\left(n+1\right)},
\]
where $a\geq0$, and $G\left(t,x,y\right)\triangleq\left(\partial_{x}^{2}-\partial_{x}\right)BS\left(t,x,y\right)$.
\end{lem}

\begin{proof}
By means of simple computations we know that
\[
G\left(s,X_{s}+a,v_{s}\right)=Ke^{-r\left(T-s\right)}\phi\left(X_{s}+a-\mu_{-},v_{s}\sqrt{T-s}\right),
\]
where $\mu_{-}\triangleq\ln\left(K\right)-\left(r-v_{s}^{2}/2\right)\left(T-s\right)$
and $\phi$ is the normal density distribution function. Notice that
$v_{s}\sqrt{T-s}=\sqrt{\int_{s}^{T}\mathbb{E}_{s}\left[\sigma_{\theta}^{2}\right]d\theta}$.
The properties of $\phi$ under the derivation sign, allow us to write
the following equality
\[
\partial_{x}^{n}G\left(s,X_{s}+a,v_{s}\right)=\left(-1\right)^{n}Ke^{-r\left(T-s\right)}\partial_{\mu_{-}}^{n}\phi\left(X_{s}+a-\mu_{-},v_{s}\sqrt{T-s}\right).
\]
Now we consider the conditional expectation with respect to the filtration
$\mathcal{G}_{t}$ that allow us to know the trajectories of the instantaneous
variance and the jump process up to time $T$. Therefore, we can write
the following,
\begin{equation}
\mathbb{E}\left[\partial_{x}^{n}G\left(s,X_{s}+a,v_{s}\right)\mid\mathcal{G}_{t}\right]=\left(-1\right)^{n}Ke^{-r\left(T-s\right)}\partial_{\mu_{-}}^{n}\mathbb{E}\left[\phi\left(X_{s}+a-\mu_{-},v_{s}\sqrt{T-s}\right)\mid\mathcal{G}_{t}\right].\label{eq: expect_npartder}
\end{equation}
The law of $X_{s}$ given $\mathcal{G}_{t}$ is a normal random variable
with mean
\[
\tilde{\mu}\triangleq X_{t}+\int_{t}^{s}\left(r-\frac{1}{2}\sigma_{\theta}^{2}\right)d\theta+\rho_{2}\eta\left(\tilde{J}_{s}-\tilde{J}_{t}\right)-\zeta\left(\rho_{2},\eta\right)\left(s-t\right)+\rho_{1}\int_{t}^{s}\sigma_{\theta}dW_{\theta},
\]
and variance $\tilde{\Sigma}\triangleq\left(1-\rho_{1}^{2}\right)\int_{t}^{s}\sigma_{\theta}^{2}d\theta=\left(1-\rho_{1}^{2}\right)\int_{t}^{s}\mathbb{E}_{s}\left[\sigma_{\theta}^{2}\right]d\theta$.
Now, we can compute the conditional expectation as follows,
\begin{align*}
 & \mathbb{E}\left[\phi\left(X_{s}+a-\mu_{-},v_{s}\sqrt{T-s}\right)\mid\mathcal{G}_{t}\right]\\
 & \qquad=\int_{\mathbb{R}}\phi\left(x+a-\mu_{-},v_{s}\sqrt{T-s}\right)f_{X}\left(x\right)dx\\
 & \qquad=\int_{\mathbb{R}}\phi\left(x+a-\mu_{-},v_{s}\sqrt{T-s}\right)\phi\left(x-\tilde{\mu},\tilde{\Sigma}\right)dx\\
 & \qquad=\int_{\mathbb{R}}\left[\frac{1}{v_{s}\sqrt{2\pi\left(T-s\right)}}\exp\left\{ -\frac{1}{2}\left(\frac{x+a-\mu_{-}}{v_{s}\sqrt{T-s}}\right)^{2}\right\} \right]\left[\frac{1}{\sqrt{2\pi\tilde{\Sigma}}}\exp\left\{ -\frac{1}{2}\left(\frac{x-\tilde{\mu}}{\sqrt{\tilde{\Sigma}}}\right)^{2}\right\} \right]dx.
\end{align*}
We know that the product of two Gaussian probability density functions,
results in the expression given by equation $\left(\ref{eq: Gaussian_PDF_Product}\right)$
from proposition $\left(\ref{prop: GaussianPDF_Product}\right)$ in
the Appendix \ref{sec:Appendix}. Therefore, we can rewrite the previous
conditional expectation as
\begin{align*}
 & \mathbb{E}\left[\phi\left(X_{s}+a-\mu_{-},v_{s}\sqrt{T-s}\right)\mid\mathcal{G}_{t}\right]\\
 & \qquad=\frac{1}{\sqrt{2\pi\left(\tilde{\Sigma}+v_{s}^{2}\left(T-s\right)\right)}}\exp\left\{ -\frac{\left(\tilde{\mu}+a-\mu_{-}\right)^{2}}{2\left(\tilde{\Sigma}+v_{s}^{2}\left(T-s\right)\right)}\right\} \\
 & \qquad\qquad\times\int_{\mathbb{R}}\left[\frac{1}{\sqrt{2\pi\frac{\tilde{\Sigma}v_{s}^{2}\left(T-s\right)}{\tilde{\Sigma}+v_{s}^{2}\left(T-s\right)}}}\exp\left\{ \frac{-\left(x-\frac{\tilde{\mu}v_{s}^{2}\left(T-s\right)+\left(\mu_{-}-a\right)\tilde{\Sigma}}{\tilde{\Sigma}+v_{s}^{2}\left(T-s\right)}\right)^{2}}{2\frac{\tilde{\Sigma}v_{s}^{2}\left(T-s\right)}{\tilde{\Sigma}+v_{s}^{2}\left(T-s\right)}}\right\} \right]dx\\
 & \qquad=\frac{1}{\sqrt{2\pi\left(\tilde{\Sigma}+v_{s}^{2}\left(T-s\right)\right)}}\exp\left\{ -\frac{\left(\tilde{\mu}+a-\mu_{-}\right)^{2}}{2\left(\tilde{\Sigma}+v_{s}^{2}\left(T-s\right)\right)}\right\} \\
 & =\phi\left(\tilde{\mu}+a-\mu_{-},\sqrt{\tilde{\Sigma}+v_{s}^{2}\left(T-s\right)}\right).
\end{align*}
The last equality results from the fact that the integral of the Gaussian
density equals to one.

Putting this result in $\left(\ref{eq: expect_npartder}\right)$,
we have 
\begin{align*}
\left|\mathbb{E}\left[\partial_{x}^{n}G\left(s,X_{s},v_{s}\right)\mid\mathcal{G}_{t}\right]\right| & =\left|\left(-1\right)^{n}Ke^{-r\left(T-s\right)}\partial_{\mu_{-}}^{n}\phi\left(\tilde{\mu}+a-\mu_{-},\sqrt{\tilde{\Sigma}+v_{s}^{2}\left(T-s\right)}\right)\right|.
\end{align*}
Now, notice that $\left|\partial_{x}^{n}\phi\left(x,\sigma\right)\right|\leq\left|\frac{x^{n}}{\sigma^{2n+1}}e^{-\frac{x^{2}}{2\sigma^{2}}}\right|$,
for all $x\in\mathbb{R}$, $\sigma\in\mathbb{R}^{+}$ and $n\geq1$.
Let $C\left(n\right)=n+1$ be a positive constant, and $d=\frac{1}{\sigma^{2}}$.
Then, trivially the following holds
\begin{align*}
\left|\partial_{x}^{n}\phi\left(x,\sigma\right)\right| & \leq x^{C}e^{-dx^{2}}.
\end{align*}
Note that the function $\psi\left(x\right)=x^{C}e^{-dx^{2}}$ has
a global maximum at $x=\pm\sqrt{\frac{C}{2d}}$ and therefore, $\left|\psi\left(x\right)\right|\leq\psi\left(\sqrt{\frac{C}{2d}}\right)\leq C\sigma^{\left(n+1\right)}$.
Therefore, since $\sigma=\sqrt{\tilde{\Sigma}+v_{s}^{2}\left(T-s\right)}$,
we can write
\begin{align*}
 & \left|\mathbb{E}\left[\partial_{x}^{n}G\left(s,X_{s},v_{s}\right)\mid\mathcal{G}_{t}\right]\right|\\
 & \leq C\left(\tilde{\Sigma}+v_{s}^{2}\left(T-s\right)\right)^{-\frac{1}{2}\left(n+1\right)}\\
 & \leq C\left(\left(1-\rho_{1}^{2}\right)\int_{t}^{s}\mathbb{E}_{s}\left[\sigma_{\theta}^{2}\right]d\theta+\int_{s}^{T}\mathbb{E}_{s}\left[\sigma_{\theta}^{2}\right]d\theta\right)^{-\frac{1}{2}\left(n+1\right)}\\
 & \leq C\left(\left(1-\rho_{1}^{2}\right)\int_{t}^{T}\mathbb{E}_{s}\left[\sigma_{\theta}^{2}\right]d\theta-\left(1-\rho_{1}^{2}\right)\int_{s}^{T}\mathbb{E}_{s}\left[\sigma_{\theta}^{2}\right]d\theta+\int_{s}^{T}\mathbb{E}_{s}\left[\sigma_{\theta}^{2}\right]d\theta\right)^{-\frac{1}{2}\left(n+1\right)}\\
 & \leq C\left(\left(1-\rho_{1}^{2}\right)\int_{t}^{T}\mathbb{E}_{s}\left[\sigma_{\theta}^{2}\right]d\theta+\rho_{1}^{2}\int_{s}^{T}\mathbb{E}_{s}\left[\sigma_{\theta}^{2}\right]d\theta\right)^{-\frac{1}{2}\left(n+1\right)}\\
 & \leq C\left(\int_{s}^{T}\mathbb{E}_{s}\left[\sigma_{\theta}^{2}\right]d\theta\right)^{-\frac{1}{2}\left(n+1\right)}.
\end{align*}
 
\end{proof}
\begin{lem}
\label{lem: LowerBound_IntExpectedVol}Assume that $2k\theta>\nu^{2}$
and let $\varphi\left(t\right)\triangleq\int_{t}^{T}e^{-\kappa\left(z-t\right)}dz=\frac{1}{\kappa}\left(1-e^{-\kappa\left(T-t\right)}\right)$.
Then, for all $0\leq s<t\leq T$
\end{lem}

\begin{enumerate}
\item $\left(i\right)$ $\int_{s}^{T}\mathbb{E}_{s}\left[\bar{\sigma}_{u}^{2}\right]du\geq\bar{\sigma}_{s}^{2}\varphi\left(s\right)$,
\item $\left(ii\right)$ $\int_{s}^{T}\mathbb{E}_{s}\left[\bar{\sigma}_{u}^{2}\right]du\geq\frac{\theta\kappa}{2}\varphi\left(s\right)^{2}$.
\end{enumerate}
\begin{proof}
From statement 2 in Lemma $\left(\ref{lem: LD_derivatives}\right)$
we know that $\mathbb{E}_{t}\left[\bar{\sigma}_{s}^{2}\right]=\bar{\sigma}_{t}^{2}e^{-\kappa\left(s-t\right)}+\theta\left(1-e^{-\kappa\left(s-t\right)}\right)$.
Now $\left(i\right)$ results from lower bounding the conditional
expectation by considering only the first term of the equality and
integrating in the interval $\left[s,T\right]$ as follows,
\begin{align*}
\int_{s}^{T}\mathbb{E}_{s}\left[\bar{\sigma}_{u}^{2}\right]du & \geq\int_{s}^{T}\bar{\sigma}_{s}^{2}e^{-\kappa\left(u-s\right)}du\geq\frac{\bar{\sigma}_{s}^{2}}{\kappa}\left(1-e^{-\kappa\left(T-s\right)}\right).
\end{align*}

In order to prove $\left(ii\right)$, we lower bound $\mathbb{E}_{t}\left[\bar{\sigma}_{s}^{2}\right]\geq\theta\kappa\varphi\left(s\right)$
and integrate in the interval $\left[s,T\right]$. Therefore, we can
write
\begin{align*}
\int_{s}^{T}\mathbb{E}_{s}\left[\bar{\sigma}_{u}^{2}\right] & \geq\frac{\theta\kappa}{2}\varphi\left(s\right)^{2}.
\end{align*}
\end{proof}
\begin{lem}
\label{lem: fwdvar_boundedderivatives} Let $g\left(t,m,y\right)\triangleq\sqrt{\frac{m-y}{T-t}}$.
Then, the following inequality holds,
\[
\left|\partial_{m}^{2}g\left(t,M_{t},Y_{t}\right)\right|\leq\frac{\left(\bar{\sigma}_{t}^{2}\varphi\left(t\right)\right)^{-\frac{3}{2}}}{4\sqrt{T-t}}.
\]
\end{lem}

\begin{proof}
By a simple calculation, we know that
\[
\partial_{m}^{2}g\left(t,m,y\right)=\frac{-1}{4\left(T-t\right)^{2}}\left(\frac{m-y}{T-t}\right)^{-\frac{3}{2}}.
\]
Recalling the definition of processes $M_{t}$ and $Y_{t}$ , the
following holds
\begin{align*}
\left|\partial_{m}^{2}g\left(t,M_{t},Y_{t}\right)\right| & =\left|\frac{-1}{4\left(T-t\right)^{2}}\left(\frac{M_{t}-Y_{t}}{T-t}\right)^{-\frac{3}{2}}\right|\\
 & =\frac{1}{4\sqrt{T-t}}\left(M_{t}-Y_{t}\right)^{-\frac{3}{2}}\\
 & =\frac{1}{4\sqrt{T-t}}\left(\int_{t}^{T}\mathbb{E}_{t}\left[\sigma_{s}^{2}\right]ds\right)^{-\frac{3}{2}}.
\end{align*}
From Lemma $\left(\ref{lem: LowerBound_IntExpectedVol}\right)$ $\left(i\right)$,
we know that $\int_{t}^{T}\mathbb{E}_{t}\left[\sigma_{u}^{2}\right]du\geq\bar{\sigma}_{t}^{2}\varphi\left(t\right)$,
finishing the proof.
\end{proof}
\begin{thm}
\label{thm: FirstOrder_ApproxFormula}(1st order approximation formula).
Fix $T>0$. Assume the model \ref{eq: fHestonJump_Volatility}, where
the volatility process $\sigma=\left\{ \sigma_{s},s\in\left[0,T\right]\right\} $
satisfies the conditions $2k\theta>\nu^{2}$ and $\left(1-c_{1}-c_{2}\frac{T^{\alpha}}{\alpha\Gamma(\alpha)}\right)>C$
for some positive constant $C$. Then 
\begin{align*}
V_{t} & \triangleq e^{-r\left(T-t\right)}\mathbb{E}_{t}\left[\left(e^{X_{T}}-K\right)^{+}\right]\\
 & =BS\left(t,X_{t},v_{t}\right)+\frac{\rho_{1}}{2}\Lambda\Gamma BS\left(t,X_{t},v_{t}\right)L\left[W,M^{c}\right]_{t}\\
 & \quad+\mathbb{E}_{t}\left[\int_{t}^{T}\int_{0}^{\infty}e^{-r\left(s-t\right)}\left(1+\frac{\rho_{1}}{2}L\left[W,M^{c}\right]_{s}\right)\Delta_{x}^{2}\Lambda\Gamma BS\left(s,X_{s-},v_{s-}\right)\ell\left(dz\right)ds\right]\\
 & \quad+\epsilon_{t},
\end{align*}
where $\epsilon_{t}$ is the error term and satisfies
\[
\left|\epsilon_{t}\right|\in\mathcal{O}\left(\nu^{2}+\eta^{2}\right).
\]
\end{thm}

\begin{proof}
This proof relies on applying Corollary \ref{cor: Funct_DGV_Relationship}
iteratively to the different terms appearing in the call price formula
given by equation $\left(\ref{eq: BS expansion formula}\right)$ in
Theorem \ref{thm: BS expansion formula}. This way, the resulting
formula will only contain terms of order $\mathcal{O}\left(\nu^{2}+\eta^{2}\right)$
which will be incorporated into the error term. 

Note that we will omit the term $-\zeta\left(\rho_{2},\eta\right)\mathbb{E}_{t}\left[\int_{t}^{T}e^{-rs}\Lambda A\left(s,X_{s},v_{s}\right)B_{s}ds\right]$
in the application of Corollary \ref{cor: Funct_DGV_Relationship}
and treat it as part of the error term in the approximating formula.
Since $\mathbb{E}_{t}\left[\left|\int_{t}^{T}e^{-rs}\Lambda A\left(s,X_{s},v_{s}\right)B_{s}ds\right|\right]<+\infty$
and from Lemma \ref{lem: Discontinuous parts} we know that
\begin{align*}
\zeta\left(\rho_{2},\eta\right) & =\int_{0}^{\infty}\left(e^{\rho_{2}\eta z}-1-\rho_{2}\eta z\right)\ell\left(dz\right)\\
 & =\rho_{2}^{2}\eta^{2}\int_{0}^{\infty}\int_{0}^{1}z^{2}e^{\lambda\rho_{2}\eta z}\left(1-\lambda\right)d\lambda\ell\left(dz\right)\leq\rho_{2}^{2}\eta^{2}\int_{0}^{\infty}z^{2}\ell\left(dz\right),
\end{align*}
where we have used that \c{ }$\rho_{2}\leq0$ so $e^{\lambda\rho_{2}\eta z}\left(1-\lambda\right)\leq1,$
and we also have that$\int_{0}^{\infty}z^{2}\ell\left(dz\right)<\infty$. 
\end{proof}
\begin{itemize}
\item Step 1: Applying Corollary $\left(\ref{cor: Funct_DGV_Relationship}\right)$
to term $\left(I\right)$ in equation $\left(\ref{eq: BS expansion formula}\right)$
with $A\left(t,X_{t},v_{t}\right)=\Lambda\Gamma BS\left(t,X_{t},v_{t}\right)$
and $B_{t}=\frac{\rho_{1}}{2}L\left[W,M^{c}\right]_{t}$ and recalling
that $B_{T}=0$ by definition, this gives
\begin{align}
\left(I\right) & =\frac{\rho_{1}}{2}\Lambda\Gamma BS\left(t,X_{t},v_{t}\right)L\left[W,M^{c}\right]_{t}\label{eq: term(I)_expansion}\\
 & \quad+\frac{\rho_{1}^{2}}{4}\mathbb{E}_{t}\left[\int_{t}^{T}e^{-r\left(s-t\right)}\Lambda^{2}\Gamma^{2}BS\left(s,X_{s},v_{s}\right)L\left[W,M^{c}\right]_{s}\sigma_{s}d\left[W,M^{c}\right]_{s}\right]\nonumber \\
 & \quad+\frac{\rho_{1}}{16}\mathbb{E}_{t}\left[\int_{t}^{T}e^{-r\left(s-t\right)}\Lambda\Gamma^{3}BS\left(s,X_{s},v_{s}\right)L\left[W,M^{c}\right]_{s}d\left[M^{c},M^{c}\right]_{s}\right]\nonumber \\
 & \quad+\frac{\rho_{1}^{2}}{2}\mathbb{E}_{t}\left[\int_{t}^{T}e^{-r\left(s-t\right)}\Lambda^{2}\Gamma BS\left(s,X_{s},v_{s}\right)\sigma_{s}d\left[W,L\left[W,M^{c}\right]\right]_{s}\right]\nonumber \\
 & \quad+\frac{\rho_{1}}{4}\mathbb{E}_{t}\left[\int_{t}^{T}e^{-r\left(s-t\right)}\Lambda\Gamma^{2}BS\left(s,X_{s},v_{s}\right)d\left[M^{c},L\left[W,M^{c}\right]\right]_{s}\right]\nonumber \\
 & \quad+\frac{\rho_{1}}{2}\mathbb{E}_{t}\left[\int_{t}^{T}e^{-r\left(s-t\right)}\Lambda\Gamma^{2}BS\left(s,X_{s},v_{s}\right)L\left[W,M^{c}\right]_{s}\right.\nonumber \\
 & \qquad\qquad\qquad\times\left.v_{s}\left(T-s\right)\int_{0}^{\infty}\Delta_{m}^{2}g\left(s,M_{s-},Y_{s-}\right)\ell\left(dz\right)ds\right]\nonumber \\
 & \quad+\frac{\rho_{1}}{2}\mathbb{E}_{t}\left[\int_{t}^{T}\int_{0}^{\infty}e^{-r\left(s-t\right)}L\left[W,M^{c}\right]_{s}\right.\nonumber \\
 & \qquad\qquad\qquad\times\left.\left[\Delta_{x}^{2}\Lambda\Gamma BS\left(s,X_{s-},v_{s-}\right)+\Delta_{y}^{2}\Lambda\Gamma BS\left(s,X_{s-},v_{s-}\right)\right]\ell(dz)ds\right]\\
 & \triangleq\left(I.I\right)+\left(I.II\right)+\ldots+\left(I.VII\right).
\end{align}
Notice also, that we can apply Lemma \ref{lem: LD_derivatives} since
we are working under the fractional Heston model with jumps. Therefore
the previous equation can be rewritten as
\begin{align*}
\left(I\right) & =\frac{\rho_{1}\nu}{2}\Lambda\Gamma BS\left(t,X_{t},v_{t}\right)\left(\int_{t}^{T}A(T,s)\mathbb{E}_{t}\left[\sqrt{\sigma_{s}^{2}\bar{\sigma}_{s}^{2}}\right]ds\right)\\
 & +\frac{\rho_{1}^{2}\nu^{2}}{4}\mathbb{E}_{t}\left[\int_{t}^{T}e^{-r\left(s-t\right)}\Lambda^{2}\Gamma^{2}BS\left(s,X_{s},v_{s}\right)\left(\int_{s}^{T}A(T,z)\mathbb{E}_{s}\left[\sqrt{\sigma_{z}^{2}\bar{\sigma}_{z}^{2}}\right]dz\right)A\left(T,s\right)\sigma_{s}\sqrt{\bar{\sigma}_{s}^{2}}ds\right]\\
 & +\frac{\rho_{1}\nu^{3}}{16}\mathbb{E}_{t}\left[\int_{t}^{T}e^{-r\left(s-t\right)}\Lambda\Gamma^{3}BS\left(s,X_{s},v_{s}\right)\left(\int_{s}^{T}A(T,z)\mathbb{E}_{s}\left[\sqrt{\sigma_{z}^{2}\bar{\sigma}_{z}^{2}}\right]dz\right)A^{2}(T,s)\bar{\sigma}_{s}^{2}ds\right]\\
 & +\frac{\rho_{1}^{2}\nu^{2}}{2}\mathbb{E}_{t}\left[\int_{t}^{T}e^{-r\left(s-t\right)}\Lambda^{2}\Gamma BS\left(s,X_{s},v_{s}\right)\left(\int_{s}^{T}A(T,z)e^{-\kappa\left(z-t\right)}dz\right)\sigma_{s}\bar{\sigma}_{s}ds\right]\\
 & +\frac{\rho_{1}\nu^{3}}{4}\mathbb{E}_{t}\left[\int_{t}^{T}e^{-r\left(s-t\right)}\Lambda\Gamma^{2}BS\left(s,X_{s},v_{s}\right)\left(\int_{s}^{T}A(T,z)e^{-\kappa\left(z-t\right)}dz\right)A(T,s)\bar{\sigma}_{s}^{2}ds\right]\\
 & +\frac{\rho_{1}\nu}{2}\mathbb{E}_{t}\left[\int_{t}^{T}e^{-r\left(s-t\right)}\Lambda\Gamma^{2}BS\left(s,X_{s},v_{s}\right)\left(\int_{s}^{T}A(T,u)\mathbb{E}_{s}\left[\sqrt{\sigma_{u}^{2}\bar{\sigma}_{u}^{2}}\right]du\right)\right.\\
 & \qquad\qquad\times\left.v_{s}\left(T-s\right)\int_{0}^{\infty}\Delta_{m}^{2}g\left(s,M_{s-},Y_{s-}\right)\ell\left(dz\right)ds\right]\\
 & +\frac{\rho_{1}\nu}{2}\mathbb{E}_{t}\left[\int_{t}^{T}\int_{0}^{\infty}e^{-r\left(s-t\right)}\left(\int_{s}^{T}A(T,u)\mathbb{E}_{s}\left[\sqrt{\sigma_{u}^{2}\bar{\sigma}_{u}^{2}}\right]du\right)\right.\\
 & \qquad\qquad\times\left.\left[\Delta_{x}^{2}\Lambda\Gamma BS\left(s,X_{s-},v_{s-}\right)+\Delta_{y}^{2}\Lambda\Gamma BS\left(s,X_{s-},v_{s-}\right)\right]\ell(dz)ds\right].
\end{align*}
The terms $\left(I.II\right)\ldots\left(I.VII\right)$, belong to
the error term $\mathcal{O}\left(\nu^{2}+\eta^{2}\right)$. We will
prove the previous statement for the terms $\left(I.III\right)$ and
$\left(I.V\right)$, the proof for the rest of the terms is analogous.
\begin{align*}
 & \left(I.III\right)+\left(I.V\right)\\
 & \quad=\frac{\rho_{1}\nu^{3}}{16}\mathbb{E}_{t}\left[\int_{t}^{T}e^{-r\left(s-t\right)}\Lambda\Gamma^{3}BS\left(s,X_{s},v_{s}\right)\left(\int_{s}^{T}A(T,z)\mathbb{E}_{s}\left[\sqrt{\sigma_{z}^{2}\bar{\sigma}_{z}^{2}}\right]dz\right)A^{2}(T,s)\bar{\sigma}_{s}^{2}ds\right]\\
 & \quad+\frac{\rho_{1}\nu^{3}}{4}\mathbb{E}_{t}\left[\int_{t}^{T}e^{-r\left(s-t\right)}\Lambda\Gamma^{2}BS\left(s,X_{s},v_{s}\right)\left(\int_{s}^{T}A(T,z)e^{-\kappa\left(z-t\right)}dz\right)A(T,s)\bar{\sigma}_{s}^{2}ds\right].
\end{align*}
Using the fact that $A\left(T,z\right)$ is a decreasing function,
defining $a_{s}\triangleq v_{s}\sqrt{T-s}$, and using the inequality
from Lemma $\left(\ref{lem: BS_BoundedDerivatives}\right)$, we can
upper bound the previous expression by
\begin{align*}
 & \left|\left(I.III\right)+\left(I.V\right)\right|\\
 & \quad\leq C\frac{\rho_{1}\nu^{3}}{16}\mathbb{E}_{t}\left[\int_{t}^{T}e^{-r\left(s-t\right)}\left(\frac{1}{a_{s}^{6}}+\frac{2}{a_{s}^{5}}+\frac{1}{a_{s}^{4}}\right)\left(\int_{s}^{T}\mathbb{E}_{s}\left[\sqrt{\sigma_{z}^{2}\bar{\sigma}_{z}^{2}}\right]dz\right)A^{3}(T,s)\bar{\sigma}_{s}^{2}ds\right]\\
 & \quad+C\frac{\rho_{1}\nu^{3}}{4}\mathbb{E}_{t}\left[\int_{t}^{T}e^{-r\left(s-t\right)}\left(\frac{1}{a_{s}^{4}}+\frac{1}{a_{s}^{3}}\right)A^{3}(T,s)\bar{\sigma}_{s}^{2}ds\right].
\end{align*}
Remember that it follows from Lemma $\left(\ref{lem: LD_derivatives}\right)$
that $\int_{s}^{T}\mathbb{E}_{s}\left[\sqrt{\sigma_{z}^{2}\bar{\sigma}_{z}^{2}}\right]dz\leq\int_{s}^{T}\mathbb{E}_{s}\left[\bar{\sigma}_{z}^{2}\right]dz+C\left(T-s\right)\nu^{2}=a_{s}^{2}+C\left(T-s\right)\nu^{2}$,
and from Lemma $\left(\ref{lem: LowerBound_IntExpectedVol}\right)$
$\left(i\right)$, that $\frac{a_{s}^{2}}{\varphi\left(s\right)}\geq\bar{\sigma}_{s}^{2}$.
Hence
\begin{align*}
 & \left|\left(I.III\right)+\left(I.V\right)\right|\\
 & \leq C\frac{\rho_{1}\nu^{3}}{16}\mathbb{E}_{t}\left[\int_{t}^{T}e^{-r\left(s-t\right)}\left(\frac{1}{a_{s}^{6}}+\frac{2}{a_{s}^{5}}+\frac{1}{a_{s}^{4}}\right)\frac{a_{s}^{4}}{\varphi\left(s\right)}A^{3}(T,s)ds\right]\\
 & \quad+C\frac{\rho_{1}\nu^{5}}{16}\mathbb{E}_{t}\left[\int_{t}^{T}e^{-r\left(s-t\right)}\left(\frac{1}{a_{s}^{6}}+\frac{2}{a_{s}^{5}}+\frac{1}{a_{s}^{4}}\right)\frac{a_{s}^{2}}{\varphi\left(s\right)}C\left(T-s\right)A^{3}(T,s)ds\right]\\
 & \quad+C\frac{\rho_{1}\nu^{3}}{16}\mathbb{E}_{t}\left[\int_{t}^{T}e^{-r\left(s-t\right)}\left(\frac{4}{a_{s}^{4}}+\frac{4}{a_{s}^{3}}\right)A^{3}(T,s)\frac{a_{s}^{2}}{\varphi\left(s\right)}ds\right]\\
 & \leq C\frac{\rho_{1}\nu^{3}}{16}\mathbb{E}_{t}\left[\int_{t}^{T}e^{-r\left(s-t\right)}\left(\frac{5}{a_{s}^{2}}+\frac{6}{a_{s}}+1\right)\frac{A^{3}(T,s)}{\varphi\left(s\right)}ds\right]\\
 & \quad+C\frac{\rho_{1}\nu^{5}}{16}\mathbb{E}_{t}\left[\int_{t}^{T}e^{-r\left(s-t\right)}\left(\frac{1}{a_{s}^{6}}+\frac{2}{a_{s}^{5}}+\frac{1}{a_{s}^{4}}\right)\frac{a_{s}^{2}}{\varphi\left(s\right)}C\left(T-s\right)A^{3}(T,s)ds\right].
\end{align*}
From $\left(ii\right)$ in Lemma $\left(\ref{lem: LowerBound_IntExpectedVol}\right)$,
we know that $a_{s}\geq\sqrt{\frac{\theta\kappa}{2}}\varphi\left(s\right)$.
Note also that $\varphi\left(s\right)\leq\frac{1}{\kappa}$ for all
$s\in\left[0,T\right]$. Taking into account that $A\left(T,s\right)\leq\varphi\left(s\right)\left(\frac{c_{2}}{\Gamma\left(\alpha\right)}\left(T-t\right)^{\alpha}+c_{1}\right)$,
we can finally upper bound the previous sum as follows
\[
\left|\left(I.III\right)+\left(I.V\right)\right|\leq C\frac{\rho_{1}\nu^{3}}{16}\mathbb{E}_{t}\left[\int_{t}^{T}e^{-r\left(s-t\right)}ds\right].
\]
 The same reasoning applies to obtain an upper bound of terms $\left(I.II\right)$,
$\left(I.IV\right)$. Notice the term $\left(I.I\right)$, depends
linearly on $\nu$ and therefore, it is part of the first order approximation
formula. We will now provide upper bounds for the discontinuous term
$\left(I.VI\right)$. We start applying Lemma $\left(\ref{lem: Discontinuous parts}\right)$
to the functions 
\begin{align*}
G_{1}\left(x\right) & =g\left(s,x,Y_{s-}\right),\\
\Delta^{2}G_{1}\left(M_{s-},c_{3}\eta z\right) & =c_{3}^{2}\eta^{2}z^{2}\int_{0}^{1}\partial_{m}^{2}g\left(s,M_{s-}+\lambda c_{3}\eta z,Y_{s-}\right)\left(1-\lambda\right)d\lambda,\\
G_{2}\left(x\right) & =\Lambda\Gamma BS\left(s,X_{s-},x\right)\\
\Delta^{2}G_{2}\left(v_{s-},c_{3}\eta z\right) & =c_{3}^{2}\eta^{2}z^{2}\int_{0}^{1}\partial_{v}^{2}\Lambda\Gamma BS\left(s,X_{s-},v_{s-}+\lambda c_{3}\eta z\right)\left(1-\lambda\right)d\lambda.
\end{align*}
Given that the terms $\Delta^{2}G_{i}$ are proportional to $\eta^{2}$,
all we need to prove is that the integrals in the term $\left(I.VI\right)$
have an upper bound, in order to properly justify that they belong
to the error term. 
\begin{itemize}
\item Using the inequality from Lemma $\left(\ref{lem: BS_BoundedDerivatives}\right)$,
we can upper bound the term $\left(I.VI\right)$ as follows
\begin{align*}
\left|\left(I.VI\right)\right| & \leq C\frac{\rho_{1}}{2}c_{3}^{2}\eta^{2}z^{2}\mathbb{E}_{t}\left[\int_{t}^{T}e^{-r\left(s-t\right)}\left(\frac{1}{a_{s}^{4}}+\frac{1}{a_{s}^{3}}\right)L\left[W,M^{c}\right]_{s}\right.\\
 & \qquad\times\left.v_{s}\left(T-s\right)\int_{0}^{\infty}\int_{0}^{1}\partial_{m}^{2}g\left(s,M_{s-}+\lambda c_{3}\eta z,Y_{s-}\right)\left(1-\lambda\right)d\lambda\ell\left(dz\right)ds\right].
\end{align*}
From Lemma $\left(\ref{lem: fwdvar_boundedderivatives}\right)$ we
can upper bound $\left|\partial_{m}^{2}g\left(s,M_{s-}+\lambda c_{3}\eta z,Y_{s-}\right)\right|$.
We also know that $L\left[W,M^{c}\right]_{s}=\nu\int_{s}^{T}A(T,r)\mathbb{E}_{s}\left[\sqrt{\sigma_{r}^{2}\bar{\sigma}_{r}^{2}}\right]dr\leq\nu A(T,s)\int_{s}^{T}\mathbb{E}_{s}\left[\sqrt{\sigma_{r}^{2}\bar{\sigma}_{r}^{2}}\right]dr$,
as $A\left(T,t\right)$ is a decreasing function of $t\in\left[0,T\right]$.
From Lemma $\left(\ref{lem: LD_derivatives}\right)$ follows that
$\int_{s}^{T}\mathbb{E}_{s}\left[\sqrt{\sigma_{r}^{2}\bar{\sigma}_{r}^{2}}\right]dr\geq\int_{s}^{T}\mathbb{E}_{s}\left[\bar{\sigma}_{r}^{2}\right]dr\triangleq a_{s}^{2}$,
and from Lemma $\left(\ref{lem: LowerBound_IntExpectedVol}\right)$
$\left(i\right)$, that $\frac{a_{s}^{2}}{\varphi\left(s\right)}\geq\bar{\sigma}_{s}^{2}$.
Therefore we can upper bound the term $\left(I.VI\right)$ as follows
\begin{align*}
\left|\left(I.VI\right)\right| & \leq C\frac{\rho_{1}}{2}c_{3}^{2}\eta^{2}z^{2}\mathbb{E}_{t}\left[\int_{t}^{T}e^{-r\left(s-t\right)}\left(\frac{1}{a_{s}^{4}}+\frac{1}{a_{s}^{3}}\right)L\left[W,M^{c}\right]_{s}\right.\\
 & \qquad\times\left.v_{s}\left(T-s\right)\int_{0}^{\infty}\left[\int_{0}^{1}\frac{1}{4\sqrt{T-s}}\left(\bar{\sigma}_{s}^{2}\varphi\left(s\right)\right)^{-\frac{3}{2}}\left(1-\lambda\right)d\lambda\right]\ell\left(dz\right)ds\right]\\
 & \leq C\frac{\rho_{1}}{8}c_{3}^{2}\eta^{2}z^{3}\mathbb{E}_{t}\left[\int_{t}^{T}e^{-r\left(s-t\right)}\left(\frac{1}{a_{s}^{4}}+\frac{1}{a_{s}^{3}}\right)L\left[W,M^{c}\right]_{s}v_{s}\sqrt{T-s}\left(\bar{\sigma}_{s}^{2}\varphi\left(s\right)\right)^{-\frac{3}{2}}ds\right]\\
 & \leq C\frac{\rho_{1}}{8}c_{3}^{2}\eta^{2}z^{3}\mathbb{E}_{t}\left[\int_{t}^{T}e^{-r\left(s-t\right)}\left(\frac{1}{a_{s}^{3}}+\frac{1}{a_{s}^{2}}\right)L\left[W,M^{c}\right]_{s}\left(a_{s}^{2}\right)^{-\frac{3}{2}}ds\right]\\
 & \leq C\frac{\rho_{1}\nu}{8}c_{3}^{2}\eta^{2}z^{3}\mathbb{E}_{t}\left[\int_{t}^{T}e^{-r\left(s-t\right)}\left(\frac{1}{a_{s}^{3}}+\frac{1}{a_{s}^{2}}\right)\frac{a_{s}^{2}}{\varphi\left(s\right)}\left(a_{s}^{2}\right)^{-\frac{3}{2}}A(T,s)ds\right]\\
 & \leq C\frac{\rho_{1}\nu}{8}c_{3}^{2}\eta^{2}z^{3}\mathbb{E}_{t}\left[\left(\frac{c_{2}}{\Gamma\left(\alpha\right)}\left(T-t\right)^{\alpha}+c_{1}\right)\int_{t}^{T}e^{-r\left(s-t\right)}\left(\frac{1}{a_{s}^{4}}+\frac{1}{a_{s}^{3}}\right)ds\right].
\end{align*}
Again, from $\left(ii\right)$ in Lemma $\left(\ref{lem: LowerBound_IntExpectedVol}\right)$,
we know that $a_{s}\geq\sqrt{\frac{\theta\kappa}{2}}\varphi\left(s\right)$
Replacing it in the previous inequality and taking into account that
$\varphi\left(s\right)\leq\frac{1}{\kappa}$ for all $s\in\left[0,T\right]$,
proves that the term $\left(I.VI\right)\in\mathcal{O}\left(\eta^{2}\right)$.
\item We can rewrite the term $\left(I.VII\right)$ recalling $\Delta^{2}G_{2}$
and the Delta-Gamma-Vega relationship given by equation $\left(\ref{eq: Delta-Gamma-Vega_Relationship(log-price)}\right)$.
If we also make use of Lemma $\left(\ref{lem: BS_BoundedDerivatives}\right)$
and use the same bounding techniques we have already been using along
the proof, we can rewrite the term as follows,
\begin{align*}
\left(I.VII\right) & =\frac{\rho_{1}}{2}\mathbb{E}_{t}\left[\int_{t}^{T}\int_{0}^{\infty}e^{-r\left(s-t\right)}L\left[W,M^{c}\right]_{s}\Delta_{x}^{2}\Lambda\Gamma BS\left(s,X_{s-},v_{s-}\right)\ell\left(dz\right)ds\right]\\
 & +\frac{\rho_{1}c_{3}^{2}\eta^{2}}{2}\mathbb{E}_{t}\left[\int_{t}^{T}\int_{0}^{\infty}e^{-r\left(s-t\right)}L\left[W,M^{c}\right]_{s}\right.\\
 & \qquad\times\left.\left[z^{2}\int_{0}^{1}\Lambda\Gamma^{3}BS\left(s,X_{s-},v_{s-}+\lambda c_{3}\eta z\right)\left(v_{s-}+\lambda c_{3}\eta z\right)^{2}\left(T-s\right)^{2}\left(1-\lambda\right)d\lambda\right]\ell(dz)ds\right]\\
 & \leq\frac{\rho_{1}}{2}\mathbb{E}_{t}\left[\int_{t}^{T}\int_{0}^{\infty}e^{-r\left(s-t\right)}L\left[W,M^{c}\right]_{s}\Delta_{x}^{2}\Lambda\Gamma BS\left(s,X_{s-},v_{s-}\right)\ell\left(dz\right)ds\right]\\
 & +C\frac{\nu\rho_{1}c_{3}^{2}\eta^{2}}{2}\mathbb{E}_{t}\left[\int_{t}^{T}e^{-r\left(s-t\right)}\left(\frac{A(T,s)}{\varphi\left(s\right)}\right)\left(\frac{1}{a_{s}^{4}}+\frac{2}{a_{s}^{3}}+\frac{1}{a_{s}^{2}}\right)\right.\\
 & \qquad\times\left.\int_{0}^{\infty}\left[\left(v_{s-}+\lambda c_{3}\eta z\right)^{2}\left(T-s\right)^{2}\right]\ell(dz)ds\right].
\end{align*}
The bounds used in the proof of the previous term also apply here,
ending the proof that shows the term $\left(I.VII\right)$ can be
written as
\[
\left(I.VII\right)=\frac{\rho_{1}}{2}\mathbb{E}_{t}\left[\int_{t}^{T}\int_{0}^{\infty}e^{-r\left(s-t\right)}L\left[W,M^{c}\right]_{s}\Delta_{x}^{2}\Lambda\Gamma BS\left(s,X_{s-},v_{s-}\right)\ell\left(dz\right)ds\right]+\mathcal{O}\left(\eta^{2}\right)
\]
\end{itemize}
\end{itemize}
This concludes the proof to show that
\begin{align*}
\left(I\right) & =\frac{\rho_{1}\nu}{2}\Lambda\Gamma BS\left(t,X_{t},v_{t}\right)\left(\int_{t}^{T}A(T,s)\mathbb{E}_{t}\left[\sqrt{\sigma_{s}^{2}\bar{\sigma}_{s}^{2}}\right]ds\right)\\
 & +\frac{\rho_{1}\nu}{2}\mathbb{E}_{t}\left[\int_{t}^{T}\int_{0}^{\infty}e^{-r\left(s-t\right)}\left(\int_{s}^{T}A(T,u)\mathbb{E}_{s}\left[\sqrt{\sigma_{u}^{2}\bar{\sigma}_{u}^{2}}\right]du\right)\Delta_{x}^{2}\Lambda\Gamma BS\left(s,X_{s-},v_{s-}\right)\ell\left(dz\right)ds\right]\\
 & +\mathcal{O}\left(v^{2}+\eta^{2}\right).
\end{align*}

\begin{itemize}
\item Step 2: Applying Corollary \ref{cor: Funct_DGV_Relationship} to term
$\left(II\right)$ in equation $\left(\ref{eq: BS expansion formula}\right)$
with $A\left(t,X_{t},v_{t}\right)=\Gamma^{2}BS\left(t,X_{t},v_{t}\right)$
and $B_{t}=\frac{1}{8}D\left[M^{c},M^{c}\right]_{t}$ and recalling
that $B_{T}=0$ by definition, we have
\begin{align*}
\left(II\right) & \qquad=\frac{1}{8}\Gamma^{2}BS\left(t,X_{t},v_{t}\right)D\left[M^{c},M^{c}\right]_{t}\\
 & \qquad+\frac{\rho_{1}}{16}\mathbb{E}_{t}\left[\int_{t}^{T}e^{-r\left(s-t\right)}\Lambda\Gamma^{3}BS\left(s,X_{s},v_{s}\right)D\left[M^{c},M^{c}\right]_{s}\sigma_{s}d\left[W,M^{c}\right]_{s}\right]\\
 & \qquad+\frac{1}{64}\mathbb{E}_{t}\left[\int_{t}^{T}e^{-r\left(s-t\right)}\Gamma^{4}BS\left(s,X_{s},v_{s}\right)D\left[M^{c},M^{c}\right]_{s}d\left[M^{c},M^{c}\right]_{s}\right]\\
 & \qquad+\frac{\rho_{1}}{8}\mathbb{E}_{t}\left[\int_{t}^{T}e^{-r\left(s-t\right)}\Lambda\Gamma^{2}BS\left(s,X_{s},v_{s}\right)\sigma_{s}d\left[W,D\left[M^{c},M^{c}\right]\right]_{s}\right]\\
 & \qquad+\frac{1}{16}\mathbb{E}_{t}\left[\int_{t}^{T}e^{-r\left(s-t\right)}\Gamma^{3}BS\left(s,X_{s},v_{s}\right)d\left[M^{c},D\left[M^{c},M^{c}\right]\right]_{s}\right]\\
 & \qquad+\frac{1}{8}\mathbb{E}_{t}\left[\int_{t}^{T}e^{-r\left(s-t\right)}\Gamma^{3}BS\left(s,X_{s},v_{s}\right)D\left[M^{c},M^{c}\right]_{s}\right.\\
 & \qquad\qquad\qquad\times\left.v_{s}\left(T-s\right)\int_{0}^{\infty}\Delta_{m}^{2}g\left(s,M_{s-},Y_{s-}\right)\ell\left(dz\right)ds\right]\\
 & \qquad+\frac{1}{8}\mathbb{E}_{t}\left[\int_{t}^{T}\int_{0}^{\infty}e^{-r\left(s-t\right)}D\left[M^{c},M^{c}\right]_{s}\right.\\
 & \qquad\qquad\qquad\times\left.\left[\Delta_{x}^{2}\Gamma^{2}BS\left(s,X_{s-},v_{s-}\right)+\Delta_{y}^{2}\Gamma^{2}BS\left(s,X_{s-},v_{s-}\right)\right]\ell(dz)ds\right].
\end{align*}
Applying Lemma \ref{lem: LD_derivatives}, the previous equation is
rewritten as
\begin{align*}
\left(II\right) & =\frac{\nu^{2}}{8}\Gamma^{2}BS\left(t,X_{t},v_{t}\right)\int_{t}^{T}A^{2}(T,s)\mathbb{E}_{t}\left[\bar{\sigma}_{s}^{2}\right]ds\\
 & \quad+\frac{\rho_{1}\nu^{3}}{16}\mathbb{E}_{t}\left[\int_{t}^{T}e^{-r\left(s-t\right)}\Lambda\Gamma^{3}BS\left(s,X_{s},v_{s}\right)\left(\int_{s}^{T}A^{2}(T,z)\mathbb{E}_{s}\left[\bar{\sigma}_{z}^{2}\right]dz\right)A(T,s)\sigma_{s}\sqrt{\bar{\sigma}_{s}^{2}}ds\right]\\
 & \quad+\frac{\nu^{4}}{64}\mathbb{E}_{t}\left[\int_{t}^{T}e^{-r\left(s-t\right)}\Gamma^{4}BS\left(s,X_{s},v_{s}\right)\left(\int_{s}^{T}A^{2}(T,z)\mathbb{E}_{s}\left[\bar{\sigma}_{z}^{2}\right]dz\right)A^{2}(T,s)\bar{\sigma}_{s}^{2}ds\right]\\
 & \quad+\frac{\rho_{1}\nu^{3}}{8}\mathbb{E}_{t}\left[\int_{t}^{T}e^{-r\left(s-t\right)}\Lambda\Gamma^{2}BS\left(s,X_{s},v_{s}\right)\left(\int_{s}^{T}A^{2}(T,z)e^{-\kappa\left(z-t\right)}dz\right)\sigma_{s}\bar{\sigma}_{s}ds\right]\\
 & \quad+\frac{\nu^{4}}{16}\mathbb{E}_{t}\left[\int_{t}^{T}e^{-r\left(s-t\right)}\Gamma^{3}BS\left(s,X_{s},v_{s}\right)\left(\int_{s}^{T}A^{2}(T,z)e^{-\kappa\left(z-t\right)}dz\right)A(T,s)\bar{\sigma}_{s}^{2}ds\right]\\
 & \quad+\frac{\nu^{2}}{8}\mathbb{E}_{t}\left[\int_{t}^{T}e^{-r\left(s-t\right)}\Gamma^{3}BS\left(s,X_{s},v_{s}\right)\left(\int_{s}^{T}A^{2}(T,u)\mathbb{E}_{s}\left[\bar{\sigma}_{u}^{2}\right]du\right)\right.\\
 & \qquad\qquad\times\left.v_{s}\left(T-s\right)\int_{0}^{\infty}\Delta_{m}^{2}g\left(s,M_{s-},Y_{s-}\right)\ell\left(dz\right)ds\right]\\
 & \quad+\frac{\nu^{2}}{8}\mathbb{E}_{t}\left[\int_{t}^{T}\int_{0}^{\infty}e^{-r\left(s-t\right)}\left(\int_{s}^{T}A^{2}(T,u)\mathbb{E}_{s}\left[\bar{\sigma}_{u}^{2}\right]du\right)\right.\\
 & \qquad\qquad\times\left.\left[\Delta_{x}^{2}\Gamma^{2}BS\left(s,X_{s-},v_{s-}\right)+\Delta_{y}^{2}\Gamma^{2}BS\left(s,X_{s-},v_{s-}\right)\right]\ell(dz)ds\right].
\end{align*}
Observe that all the terms in $\left(II\right)$ can be incorporated
into the error term, i.e. $\left(II\right)=\mathcal{O}\left(\nu^{2}\right).$
This is due to the dependency of each term on higher order of $\nu$
and the fact that all terms can be upper bounded following the reasoning
we provided for terms $\left(I.III\right)$ and $\left(I.V\right)$.
\item Step 3: Proving the term $\left(III\right)$ in equation $\left(\ref{eq: BS expansion formula}\right)$
belongs to the error term $\mathcal{O}\left(\eta^{2}\right)$, is
analogous to the discussion we did for the term $\left(I.VI\right)$.
\item Step 4: Again, an analogous discussion as the one performed in the
study of $\left(I.VII\right)$ applies here, showing that
\[
\left(IV\right)=\mathbb{E}_{t}\left[\int_{t}^{T}\int_{0}^{\infty}e^{-r\left(s-t\right)}\Delta_{x}^{2}\Lambda\Gamma BS\left(s,X_{s-},v_{s-}\right)\ell\left(dz\right)ds\right]+\mathcal{O}\left(\eta^{2}\right).
\]
 
\end{itemize}

\section{\label{sec:Appendix}Appendix}

In this appendix we gather additional technical lemmas.
\begin{lem}
\label{lem: Discontinuous parts}Let $G\in C^{2}\left(\mathbb{R}\right)$
and consider the expression $\Delta^{2}G\left(x,h\right)$ defined
as
\[
\Delta^{2}G\left(x,h\right)\triangleq G\left(x+h\right)-G\left(x\right)-hG^{\prime}\left(x\right).
\]
Then, the following equality holds
\begin{align}
\Delta^{2}G\left(x,h\right) & =h^{2}\int_{0}^{1}G^{\prime\prime}\left(x+\lambda h\right)\left(1-\lambda\right)d\lambda.\label{eq: Delta2_TaylorApprox}
\end{align}
\end{lem}

\begin{proof}
It is Taylor's Theorem with integral remainder.
\end{proof}
\begin{prop}
\label{prop: BlackScholes_DGV_Relationship}(Delta-Gamma-Vega Relationship)
Let 
\[
BS\left(t,x,y\right)=e^{x}\Phi\left(d_{+}\right)-e^{-r\left(T-t\right)}K\Phi\left(d_{-}\right),
\]
 where $\Phi$ denotes the cumulative distribution function of a standard
normal distribution and 
\[
d_{\pm}=\frac{x-\ln K+\left(r\pm\frac{y^{2}}{2}\right)\left(T-t\right)}{y\sqrt{T-t}};\qquad d_{+}=d_{-}+y\sqrt{T-t}.
\]
Then for every $t\in\left[0,T\right]$, the following formula, known
as the Delta-Gamma-Vega relationship, holds.
\begin{equation}
\partial_{y}BS\left(t,x,y\right)\frac{1}{y\left(T-t\right)}=\left(\partial_{xx}^{2}-\partial_{x}\right)BS\left(t,x,y\right).\label{eq: Delta-Gamma-Vega_Relationship(log-price)}
\end{equation}
\end{prop}

\begin{proof}
We start by computing the log-Delta.\footnote{We will name log-Delta the change in the option price with respect
to the change in the underlying asset log-price. This is a different
sensitivity of what is commonly referred to as the Delta, i.e. the
change in the option price with respect to the change in the underlying
asset price.}
\begin{align*}
\partial_{x}BS\left(t,x,y\right) & =e^{x}\Phi\left(d_{+}\right).
\end{align*}
Now the log-Gamma is computed as
\begin{align*}
\partial_{xx}^{2}BS\left(t,x,y\right) & =e^{x}\Phi\left(d_{+}\right)+\frac{e^{x}\phi\left(d_{+}\right)}{y\sqrt{T-t}}.
\end{align*}
Therefore, we have that 
\[
\left(\partial_{xx}^{2}-\partial_{x}\right)BS\left(t,x,y\right)=\frac{e^{x}\phi\left(d_{+}\right)}{y\sqrt{T-t}}.
\]
On the other hand, the Vega is derived as follows
\begin{align*}
\partial_{y}BS\left(t,x,y\right) & =e^{x}\phi\left(d_{+}\right)\partial_{y}d_{+}-e^{-r\left(T-t\right)}K\phi\left(d_{-}\right)\partial_{y}d_{-}\\
 & =e^{x}\phi\left(d_{+}\right)\left[\partial_{y}d_{-}+\partial_{y}\left(y\sqrt{T-t}\right)\right]-e^{-r\left(T-t\right)}K\phi\left(d_{-}\right)\partial_{y}d_{-}\\
 & =e^{x}\phi\left(d_{+}\right)\partial_{y}\left(y\sqrt{T-t}\right)\\
 & =e^{x}\phi\left(d_{+}\right)\sqrt{T-t}.
\end{align*}
The relationship follows trivially from these computations.
\end{proof}
\begin{prop}
\label{prop: GaussianPDF_Product}Let $\phi\left(x,\sigma\right)$
denote the density function of the normal law with mean zero and standard
deviation $\sigma$. Then, for $\mu_{1},\mu_{2}\in\mathbb{R}$ and
$\sigma_{1},\sigma_{2}$ strictly positive real numbers we have
\begin{equation}
\phi\left(x-\mu_{1},\sigma_{1}\right)\phi\left(x-\mu_{2},\sigma_{2}\right)=\phi\left(x-\frac{\mu_{1}\sigma_{2}^{2}+\mu_{2}\sigma_{1}^{2}}{\sigma_{1}^{2}+\sigma_{2}^{2}},\sqrt{\frac{\sigma_{1}^{2}\sigma_{2}^{2}}{\sigma_{1}^{2}+\sigma_{2}^{2}}}\right)\phi\left(\mu_{1}-\mu_{2},\sqrt{\sigma_{1}^{2}+\sigma_{2}^{2}}\right).\label{eq: Gaussian_PDF_Product}
\end{equation}
\end{prop}

\begin{proof}
This proof results from basic algebraic manipulations. Note that we
can trivially write the product of densities as follows
\begin{align*}
\phi\left(x-\mu_{1},\sigma_{1}\right)\phi\left(x-\mu_{2},\sigma_{2}\right) & =\frac{1}{\sqrt{2\pi\sigma_{1}^{2}}}\exp\left\{ -\frac{\left(x-\mu_{1}\right)^{2}}{2\sigma_{1}^{2}}\right\} \frac{1}{\sqrt{2\pi\sigma_{2}^{2}}}\exp\left\{ -\frac{\left(x-\mu_{2}\right)^{2}}{2\sigma_{2}^{2}}\right\} \\
 & =\frac{1}{\sqrt{2\pi\sigma_{1}^{2}\sigma_{2}^{2}2\pi}}\exp\left\{ -\frac{\left(x-\mu_{1}\right)^{2}}{2\sigma_{1}^{2}}-\frac{\left(x-\mu_{2}\right)^{2}}{2\sigma_{2}^{2}}\right\} \\
 & =\frac{1}{\sqrt{2\pi\sigma_{1}^{2}\sigma_{2}^{2}2\pi}}\exp\left\{ \frac{-\sigma_{2}^{2}\left(x-\mu_{1}\right)^{2}-\sigma_{1}^{2}\left(x-\mu_{2}\right)^{2}}{2\sigma_{1}^{2}\sigma_{2}^{2}}\right\} .
\end{align*}
Since $\sigma_{1}^{2}+\sigma_{2}^{2}>0$ we can rewrite the previous
expression as
\begin{align*}
\phi\left(x-\mu_{1},\sigma_{1}\right)\phi\left(x-\mu_{2},\sigma_{2}\right) & =\frac{1}{\sqrt{2\pi\frac{\sigma_{1}^{2}\sigma_{2}^{2}}{\sigma_{1}^{2}+\sigma_{2}^{2}}\left(\sigma_{1}^{2}+\sigma_{2}^{2}\right)2\pi}}\\
 & \times\exp\left\{ \frac{-\frac{\sigma_{2}^{2}}{\sigma_{1}^{2}+\sigma_{2}^{2}}\left(x-\mu_{1}\right)^{2}-\frac{\sigma_{1}^{2}}{\sigma_{1}^{2}+\sigma_{2}^{2}}\left(x-\mu_{2}\right)^{2}}{2\frac{\sigma_{1}^{2}\sigma_{2}^{2}}{\sigma_{1}^{2}+\sigma_{2}^{2}}}\right\} .
\end{align*}
Notice from the previous expression that 
\[
\sqrt{2\pi\frac{\sigma_{1}^{2}\sigma_{2}^{2}}{\sigma_{1}^{2}+\sigma_{2}^{2}}\left(\sigma_{1}^{2}+\sigma_{2}^{2}\right)2\pi}=\sqrt{2\pi\frac{\sigma_{1}^{2}\sigma_{2}^{2}}{\sigma_{1}^{2}+\sigma_{2}^{2}}}\sqrt{2\pi\left(\sigma_{1}^{2}+\sigma_{2}^{2}\right)},
\]
therefore, replacing this in the product of the two pdf's and considering
that $x=e^{\ln x}$, one obtains that
\begin{align*}
 & \phi\left(x-\mu_{1},\sigma_{1}\right)\phi\left(x-\mu_{2},\sigma_{2}\right)\\
 & =\frac{1}{\sqrt{2\pi\frac{\sigma_{1}^{2}\sigma_{2}^{2}}{\sigma_{1}^{2}+\sigma_{2}^{2}}}}\exp\left\{ \ln\frac{1}{\sqrt{2\pi\left(\sigma_{1}^{2}+\sigma_{2}^{2}\right)}}\right\} \exp\left\{ \frac{-\frac{\sigma_{2}^{2}}{\sigma_{1}^{2}+\sigma_{2}^{2}}\left(x-\mu_{1}\right)^{2}-\frac{\sigma_{1}^{2}}{\sigma_{1}^{2}+\sigma_{2}^{2}}\left(x-\mu_{2}\right)^{2}}{2\frac{\sigma_{1}^{2}\sigma_{2}^{2}}{\sigma_{1}^{2}+\sigma_{2}^{2}}}\right\} \\
 & \frac{1}{\sqrt{2\pi\frac{\sigma_{1}^{2}\sigma_{2}^{2}}{\sigma_{1}^{2}+\sigma_{2}^{2}}}}\exp\left\{ -\frac{1}{2}\ln\left(2\pi\left(\sigma_{1}^{2}+\sigma_{2}^{2}\right)\right)+\frac{-\frac{\sigma_{2}^{2}}{\sigma_{1}^{2}+\sigma_{2}^{2}}\left(x-\mu_{1}\right)^{2}-\frac{\sigma_{1}^{2}}{\sigma_{1}^{2}+\sigma_{2}^{2}}\left(x-\mu_{2}\right)^{2}}{2\frac{\sigma_{1}^{2}\sigma_{2}^{2}}{\sigma_{1}^{2}+\sigma_{2}^{2}}}\right\} .
\end{align*}
Focusing on the exponential term, we will perform some further algebraic
manipulations on it as follows. 
\begin{align*}
 & -\frac{1}{2}\ln\left(2\pi\left(\sigma_{1}^{2}+\sigma_{2}^{2}\right)\right)+\frac{-\frac{\sigma_{2}^{2}}{\sigma_{1}^{2}+\sigma_{2}^{2}}\left(x-\mu_{1}\right)^{2}-\frac{\sigma_{1}^{2}}{\sigma_{1}^{2}+\sigma_{2}^{2}}\left(x-\mu_{2}\right)^{2}}{2\frac{\sigma_{1}^{2}\sigma_{2}^{2}}{\sigma_{1}^{2}+\sigma_{2}^{2}}}\\
 & \quad=\frac{-\left(\frac{\sigma_{1}^{2}\sigma_{2}^{2}}{\sigma_{1}^{2}+\sigma_{2}^{2}}\right)\ln\left(2\pi\left(\sigma_{1}^{2}+\sigma_{2}^{2}\right)\right)-\frac{\sigma_{2}^{2}}{\sigma_{1}^{2}+\sigma_{2}^{2}}\left(x^{2}-2x\mu_{1}+\mu_{1}^{2}\right)-\frac{\sigma_{1}^{2}}{\sigma_{1}^{2}+\sigma_{2}^{2}}\left(x^{2}-2x\mu_{2}+\mu_{2}^{2}\right)}{2\frac{\sigma_{1}^{2}\sigma_{2}^{2}}{\sigma_{1}^{2}+\sigma_{2}^{2}}}\\
 & \quad=\frac{\frac{-\sigma_{1}^{2}-\sigma_{2}^{2}}{\sigma_{1}^{2}+\sigma_{2}^{2}}x^{2}-2x\left(\frac{-\mu_{1}\sigma_{2}^{2}-\mu_{2}\sigma_{1}^{2}}{\sigma_{1}^{2}+\sigma_{2}^{2}}\right)+\left(\frac{-\mu_{1}^{2}\sigma_{2}^{2}-\mu_{2}^{2}\sigma_{1}^{2}-\sigma_{1}^{2}\sigma_{2}^{2}\ln\left(2\pi\left(\sigma_{1}^{2}+\sigma_{2}^{2}\right)\right)}{\sigma_{1}^{2}+\sigma_{2}^{2}}\right)}{2\frac{\sigma_{1}^{2}\sigma_{2}^{2}}{\sigma_{1}^{2}+\sigma_{2}^{2}}}\\
 & \quad=\frac{-x^{2}+2x\left(\frac{-\mu_{1}\sigma_{2}^{2}-\mu_{2}\sigma_{1}^{2}}{\sigma_{1}^{2}+\sigma_{2}^{2}}\right)-\left(\frac{\mu_{1}^{2}\sigma_{2}^{2}+\mu_{2}^{2}\sigma_{1}^{2}+\sigma_{1}^{2}\sigma_{2}^{2}\ln\left(2\pi\left(\sigma_{1}^{2}+\sigma_{2}^{2}\right)\right)}{\sigma_{1}^{2}+\sigma_{2}^{2}}\right)}{2\frac{\sigma_{1}^{2}\sigma_{2}^{2}}{\sigma_{1}^{2}+\sigma_{2}^{2}}}.
\end{align*}
Now, one can write the previous equality as a second order polynomial
of the form $-x^{2}+2Ax-A^{2}-C=-\left(x-A\right)^{2}-C$, where $A$
and $C$ are given by
\begin{align*}
A & =\frac{\mu_{1}\sigma_{2}^{2}+\mu_{2}\sigma_{1}^{2}}{\sigma_{1}^{2}+\sigma_{2}^{2}},\\
C & =\frac{\sigma_{1}^{2}\sigma_{2}^{2}}{\left(\sigma_{1}^{2}+\sigma_{2}^{2}\right)^{2}}\left(\mu_{1}-\mu_{2}\right)^{2}+\frac{\sigma_{1}^{2}\sigma_{2}^{2}}{\sigma_{1}^{2}+\sigma_{2}^{2}}\ln\left(2\pi\left(\sigma_{1}^{2}+\sigma_{2}^{2}\right)\right).
\end{align*}
Equiped with this, we can now rewrite the product of the two Gaussian
densities as
\begin{align*}
\phi\left(x-\mu_{1},\sigma_{1}\right)\phi\left(x-\mu_{2},\sigma_{2}\right) & =\frac{1}{\sqrt{2\pi\frac{\sigma_{1}^{2}\sigma_{2}^{2}}{\sigma_{1}^{2}+\sigma_{2}^{2}}}}\exp\left\{ \frac{-\left(x-A\right)^{2}-C}{2\frac{\sigma_{1}^{2}\sigma_{2}^{2}}{\sigma_{1}^{2}+\sigma_{2}^{2}}}\right\} ,\\
 & =\phi\left(x-A,\sqrt{\frac{\sigma_{1}^{2}\sigma_{2}^{2}}{\sigma_{1}^{2}+\sigma_{2}^{2}}}\right)\exp\left\{ \frac{-C}{2\frac{\sigma_{1}^{2}\sigma_{2}^{2}}{\sigma_{1}^{2}+\sigma_{2}^{2}}}\right\} .
\end{align*}
We will focus on the second exponential term to further expand it.
\begin{align*}
\exp\left\{ \frac{-C}{2\frac{\sigma_{1}^{2}\sigma_{2}^{2}}{\sigma_{1}^{2}+\sigma_{2}^{2}}}\right\}  & =\exp\left\{ -\frac{\frac{\sigma_{1}^{2}\sigma_{2}^{2}}{\left(\sigma_{1}^{2}+\sigma_{2}^{2}\right)^{2}}\left(\mu_{1}-\mu_{2}\right)^{2}+\frac{\sigma_{1}^{2}\sigma_{2}^{2}}{\sigma_{1}^{2}+\sigma_{2}^{2}}\ln\left(2\pi\left(\sigma_{1}^{2}+\sigma_{2}^{2}\right)\right)}{2\frac{\sigma_{1}^{2}\sigma_{2}^{2}}{\sigma_{1}^{2}+\sigma_{2}^{2}}}\right\} \\
 & =\frac{1}{\sqrt{2\pi\left(\sigma_{1}^{2}+\sigma_{2}^{2}\right)}}\exp\left\{ -\frac{\left(\mu_{1}-\mu_{2}\right)^{2}}{2\left(\sigma_{1}^{2}+\sigma_{2}^{2}\right)}\right\} \\
 & =\phi\left(\mu_{1}-\mu_{2},\sqrt{\sigma_{1}^{2}+\sigma_{2}^{2}}\right).
\end{align*}
This ends the proof.
\end{proof}

\subsection{\label{subsec: Proof Lemma LD_derivatives}Proof of Lemma \ref{lem: LD_derivatives}}

We will only prove by order statements 3, 9 and 10, as the rest are
trivially deduced from the definitions and results provided in previous
sections. 
\begin{itemize}
\item In order to see that $\mathbb{E}_{t}\left[\sigma_{s}\sqrt{\bar{\sigma}_{s}^{2}}\right]=\bar{\sigma}_{t}^{2}e^{-\kappa\left(s-t\right)}+\theta\left(1-e^{-\kappa\left(s-t\right)}\right)+\mathcal{O}\left(\nu^{2}+\eta^{2}\right)$
we consider the process 
\[
u_{r}\triangleq\sqrt{\left(U_{s}+\nu Z_{r}^{s}\right)\left(U_{s}+c_{1}\nu Z_{r}^{s}+c_{2}\nu I_{0+}^{H-\frac{1}{2}}Z_{r}^{s}+c_{3}\eta J_{r}\right)},
\]
where 
\begin{align*}
U_{t} & \triangleq\theta+e^{-\kappa t}\left(\bar{\sigma}_{0}^{2}-\theta\right),\\
Z_{r}^{s} & \triangleq\int_{0}^{r}e^{-\kappa\left(s-u\right)}\sqrt{\bar{\sigma}_{u}^{2}}dW_{u},\\
I_{0+}^{H-\frac{1}{2}}Z_{r}^{s} & =\frac{1}{\Gamma\left(H-\frac{1}{2}\right)}\int_{0}^{r}Z_{u}\left(s-u\right)^{H-\frac{3}{2}}du\\
 & =\frac{1}{\Gamma\left(H-\frac{1}{2}\right)}\int_{0}^{r}\left(\int_{0}^{u}e^{-\kappa\left(u-v\right)}\sqrt{\bar{\sigma}_{v}^{2}}dW_{v}\right)\left(s-u\right)^{H-\frac{3}{2}}du,\qquad r\in\left[0,s\right].
\end{align*}
Note that $u_{s}=\sqrt{\sigma_{s}^{2}\bar{\sigma}_{s}^{2}}$, $Z_{s}^{s}=Z_{s},s\in\left[0,T\right]$
and $I_{0+}^{H-\frac{1}{2}}Z_{s}^{s}=I_{0+}^{H-\frac{1}{2}}Z_{s}$.
Applying Fubini, we can write 
\begin{align*}
I_{0+}^{H-\frac{1}{2}}Z_{r}^{s} & =\frac{1}{\Gamma\left(H-\frac{1}{2}\right)}\int_{0}^{r}\left(\int_{v}^{r}e^{-\kappa\left(u-v\right)}\left(s-u\right)^{H-\frac{3}{2}}du\right)\sqrt{\bar{\sigma}_{v}^{2}}dW_{v}\\
 & =\int_{0}^{r}\psi\left(r,s,v\right)\sqrt{\bar{\sigma}_{v}^{2}}dW_{v},
\end{align*}
where 
\begin{align*}
\psi\left(r,s,v\right) & =\frac{1}{\Gamma\left(H-\frac{1}{2}\right)}\int_{v}^{r}e^{-\kappa\left(u-v\right)}\left(s-u\right)^{H-\frac{3}{2}}du\\
 & =\frac{e^{-\kappa\left(s-v\right)}}{\Gamma\left(H-\frac{1}{2}\right)}\int_{s-v}^{s-r}e^{\kappa w}w^{H-\frac{3}{2}}dw.
\end{align*}
Next, consider the process 
\begin{align*}
\Pi_{r}^{s} & =c_{1}\nu Z_{r}^{s}+c_{2}\nu I_{0+}^{H-\frac{1}{2}}Z_{r}^{s}\\
 & =c_{1}\nu\int_{0}^{r}e^{-\kappa\left(s-u\right)}\sqrt{\bar{\sigma}_{u}^{2}}dW_{u}+c_{2}\nu\int_{0}^{r}\psi\left(r,s,u\right)\sqrt{\bar{\sigma}_{u}^{2}}dW_{u}\\
 & =\int_{0}^{r}\zeta\left(r,s,u\right)\sqrt{\bar{\sigma}_{u}^{2}}dW_{u},
\end{align*}
where 
\[
\zeta\left(r,s,u\right)=c_{1}\nu e^{-\kappa\left(s-u\right)}+c_{2}\nu\psi\left(r,s,u\right).
\]
Note that 
\[
\zeta\left(r,s,r\right)=c_{1}\nu e^{-\kappa\left(s-r\right)}+c_{2}\nu\psi\left(r,s,r\right)=c_{1}\nu e^{-\kappa\left(s-r\right)},
\]
and 
\[
\partial_{1}\zeta\left(r,s,u\right)=c_{2}\nu\partial_{1}\psi\left(r,s,u\right)=\frac{c_{2}\nu}{\Gamma\left(H-\frac{1}{2}\right)}e^{-\kappa\left(r-u\right)}\left(s-r\right)^{H-\frac{3}{2}}.
\]
Therefore, 
\begin{align*}
d\Pi_{r}^{s} & =\left(c_{1}\nu e^{-\kappa\left(s-u\right)}+c_{2}\nu\psi\left(r,s,u\right)\right)\sqrt{\bar{\sigma}_{r}^{2}}dW_{r}\\
 & +\left(\frac{c_{2}\nu}{\Gamma\left(H-\frac{1}{2}\right)}\int_{0}^{r}e^{-\kappa\left(r-u\right)}\left(s-r\right)^{H-\frac{3}{2}}\sqrt{\bar{\sigma}_{u}^{2}}dW_{u}\right)dr.
\end{align*}
Consider the following process $\Theta_{r}^{s}=\Pi_{r}^{s}+c_{3}\eta J_{r}$.
We have that 
\begin{align*}
dZ_{r}^{s} & =e^{-\kappa\left(s-r\right)}\sqrt{\bar{\sigma}_{r}^{2}}dW_{r},\\
d\Theta_{r}^{s} & =d\Pi_{r}^{s}+c_{3}\eta dJ_{r},\\
d\Theta_{r}^{s,c} & =d\Pi_{r}^{s},\\
d\langle Z^{s},Z^{s}\rangle_{r} & =e^{-2\kappa\left(s-r\right)}\bar{\sigma}_{r}^{2}dr,\\
d\langle Z^{r},\Theta^{s,c}\rangle_{r} & =\left\{ c_{1}\nu e^{-\kappa\left(s-u\right)}+c_{2}\nu\psi\left(r,s,u\right)\right\} e^{-\kappa\left(s-r\right)}\bar{\sigma}_{r}^{2}dr,\\
d\langle\Theta^{s,c},\Theta^{s,c}\rangle_{t} & =\left(c_{1}\nu e^{-\kappa\left(s-u\right)}+c_{2}\nu\psi\left(r,s,u\right)\right)^{2}\bar{\sigma}_{r}^{2}dr.
\end{align*}
Next, we can apply It\^{o} formula to $f\left(Z_{r}^{s},\Theta_{r}^{s}\right),$where
$f$ is the function defined by 
\[
f\left(x,y\right)\triangleq\sqrt{\left(U_{s}+\nu x\right)\left(U_{s}+y\right)}.
\]
We have that 
\begin{align*}
\partial_{1}f\left(x,y\right) & =\frac{\nu}{2}\left(\frac{U_{s}+y}{U_{s}+\nu x}\right)^{1/2},\qquad\partial_{2}f\left(x,y\right)=\frac{1}{2}\left(\frac{U_{s}+\nu x}{U_{s}+y}\right)^{1/2},\\
\partial_{11}f\left(x,y\right) & =-\frac{\nu^{2}}{4}\left(U_{s}+y\right)^{1/2}\left(U_{s}+\nu x\right)^{-3/2},\\
\partial_{12}f\left(x,y\right) & =\frac{\nu}{4}\left(U_{s}+y\right)^{-1/2}\left(U_{s}+\nu x\right)^{-1/2},\\
\partial_{22}f\left(x,y\right) & =-\frac{1}{4}\left(U_{s}+y\right)^{-3/2}\left(U_{s}+\nu x\right)^{1/2}.
\end{align*}
Note that 
\[
f\left(Z_{s}^{s},\Theta_{s}^{s}\right)=\sqrt{\bar{\sigma}_{s}^{2}\sigma_{s}^{2}}=\sqrt{\left(U_{s}+\nu Z_{s}\right)\left(U_{s}+c_{1}\nu Z_{s}+c_{2}\nu I_{0+}^{H-\frac{1}{2}}Z_{s}+c_{3}\eta J_{s}\right)}.
\]
Now, an application of the It\^{o} formula yields the following expression
\begin{align*}
 & f\left(Z_{s}^{s},\Theta_{s}^{s}\right)\\
 & =f\left(Z_{0}^{s},\Theta_{0}^{s}\right)+\int_{0}^{s}\partial_{1}f\left(Z_{r}^{s},\Theta_{r}^{s}\right)dZ_{r}^{s}+\int_{0}^{s}\partial_{2}f\left(Z_{r}^{s},\Theta_{r}^{s}\right)d\Theta_{r}^{s,c}\\
 & \quad+\int_{0}^{s}\frac{1}{2}\partial_{11}^{2}f\left(Z_{r}^{s},\Theta_{r}^{s}\right)d\langle Z^{s},Z^{s}\rangle_{r}+\int_{0}^{s}\partial_{12}^{2}f\left(Z_{r}^{s},\Theta_{r}^{s}\right)d\langle Z^{s},\Theta^{s,c}\rangle_{r}\\
 & \quad+\int_{0}^{s}\frac{1}{2}\partial_{22}^{2}f\left(Z_{r}^{s},\Theta_{r}^{s}\right)d\langle\Theta^{s,c},\Theta^{s,c}\rangle_{r}\\
 & \quad+\int_{0}^{s}\int_{0}^{\infty}\Delta_{y}f\left(Z_{r}^{s},\Theta_{r-}^{s}\right)\tilde{N}\left(dr,dz\right)+\int_{0}^{s}\int_{0}^{\infty}\Delta_{y}^{2}f\left(Z_{r}^{s},\Theta_{r-}^{s}\right)\ell\left(dz\right)dr\\
 & =f\left(Z_{0}^{s},\Theta_{0}^{s}\right)+\frac{\nu}{2}\int_{0}^{s}\left(\frac{U_{s}+\Theta_{r}^{s}}{U_{s}+\nu Z_{r}^{s}}\right)^{1/2}e^{-\kappa\left(s-r\right)}\sqrt{\bar{\sigma}_{r}^{2}}dW_{r}\\
 & \quad+\frac{\nu}{2}\int_{0}^{s}\left(\frac{U_{s}+\nu Z_{r}^{s}}{U_{s}+\Theta_{r}^{s}}\right)^{1/2}\left[\left(c_{1}e^{-\kappa\left(s-u\right)}+c_{2}\psi\left(r,s,u\right)\right)\sqrt{\bar{\sigma}_{r}^{2}}dW_{r}\right.\\
 & \quad\qquad\left.+\left(\frac{c_{2}}{\Gamma\left(H-\frac{1}{2}\right)}\int_{0}^{r}e^{-\kappa\left(r-u\right)}\left(s-r\right)^{H-\frac{3}{2}}\sqrt{\bar{\sigma}_{u}^{2}}dW_{u}\right)dr\right]\\
 & \quad-\frac{\nu^{2}}{8}\int_{0}^{s}\left(U_{s}+\Theta_{r}^{s}\right)^{1/2}\left(U_{s}+\nu Z_{r}^{s}\right)^{-3/2}e^{-2\kappa\left(s-r\right)}\bar{\sigma}_{r}^{2}dr\\
 & \quad+\frac{\nu^{2}}{4}\int_{0}^{s}\left(U_{s}+\Theta_{r}^{s}\right)^{-1/2}\left(U_{s}+\nu Z_{r}^{s}\right)^{-1/2}\left\{ c_{1}e^{-\kappa\left(s-u\right)}+c_{2}\psi\left(r,s,u\right)\right\} e^{-\kappa\left(s-r\right)}\bar{\sigma}_{r}^{2}dr\\
 & \quad-\frac{\nu^{2}}{8}\int_{0}^{s}\left(U_{s}+\Theta_{r}^{s}\right)^{-3/2}\left(U_{s}+\nu Z_{r}^{s}\right)^{1/2}\left(c_{1}e^{-\kappa\left(s-u\right)}+c_{2}\psi\left(r,s,u\right)\right)^{2}\bar{\sigma}_{r}^{2}dr\\
 & \quad+\int_{0}^{s}\int_{0}^{\infty}\Delta_{y}f\left(Z_{r}^{s},\Theta_{r-}^{s}\right)\tilde{N}\left(dr,dz\right)+\int_{0}^{s}\int_{0}^{\infty}\Delta_{y}^{2}f\left(Z_{r}^{s},\Theta_{r-}^{s}\right)\ell\left(dz\right)dr.
\end{align*}
Now, taking expectations and writing the terms with $\nu^{2}$ as
an error term of order $\mathcal{O}\left(\nu^{2}\right)$, we can
rewrite the previous equation as 
\begin{align*}
\mathbb{E}\left[f\left(Z_{s}^{s},\Theta_{s}^{s}\right)\right] & =\mathbb{E}\left[f\left(Z_{0}^{s},\Theta_{0}^{s}\right)\right]\\
 & \quad+\frac{c_{2}\nu}{2\Gamma\left(H-\frac{1}{2}\right)}\mathbb{E}\left[\int_{0}^{s}\int_{0}^{r}\left(\frac{U_{s}+\nu Z_{r}^{s}}{U_{s}+\Theta_{r}^{s}}\right)^{1/2}e^{-\kappa\left(r-u\right)}\left(s-r\right)^{H-\frac{3}{2}}\sqrt{\bar{\sigma}_{u}^{2}}dW_{u}dr\right]\\
 & \quad+\mathbb{E}\left[\int_{0}^{s}\int_{0}^{\infty}\Delta_{y}^{2}f\left(Z_{r}^{s},\Theta_{r-}^{s}\right)\ell\left(dz\right)dr\right]+\mathcal{O}\left(\nu^{2}\right).
\end{align*}
Again, using Fubini and noting the integrand inside the expectation
for the continuous term is square integrable, we obtain 
\begin{align*}
\mathbb{E}\left[f\left(Z_{s}^{s},\Theta_{s}^{s}\right)\right] & =\mathbb{E}\left[U_{s}\right]+\mathbb{E}\left[\int_{0}^{s}\int_{0}^{\infty}\Delta_{y}^{2}f\left(Z_{r}^{s},\Theta_{r-}^{s}\right)\ell\left(dz\right)dr\right]+\mathcal{O}\left(\nu^{2}\right),
\end{align*}
where $f\left(Z_{0}^{s},\Theta_{0}^{s}\right)=U_{s}$. Finally, it
only remains to notice that the term $\Delta_{y}^{2}f\left(Z_{r}^{s},\Theta_{r-}^{s}\right)=\partial_{22}f\left(Z_{s}^{s},\Theta_{s}^{s}\right)c_{3}^{2}\eta^{2}z^{2}\in\mathcal{O}\left(\eta^{2}\right)$,
allowing us to write 
\[
\mathbb{E}\left[f\left(Z_{s}^{s},\Theta_{s}^{s}\right)\right]=\mathbb{E}\left[\sqrt{\sigma_{s}^{2}\bar{\sigma}_{s}^{2}}\right]=\mathbb{E}\left[U_{s}\right]+\mathcal{O}\left(\nu^{2}+\eta^{2}\right).
\]
\item In order to see that $dL\left[W,M^{c}\right]_{t}=\nu^{2}\left(\int_{t}^{T}A(T,s)e^{-\kappa\left(s-t\right)}ds\right)\bar{\sigma}_{t}dW_{t}-\nu\bar{\sigma}_{t}^{2}A\left(T,t\right)dt+\mathcal{O}\left(\nu^{3}+\nu\eta^{2}\right)dt$,
we can compute the following: 
\begin{align*}
dL\left[W,M^{c}\right]_{t} & =d\left(\nu\int_{t}^{T}A(T,s)\mathbb{E}_{t}\left[\sqrt{\sigma_{s}^{2}\bar{\sigma}_{s}^{2}}\right]ds\right)\\
 & =d\left(\nu\int_{t}^{T}A(T,s)\left[\bar{\sigma}_{t}^{2}e^{-\kappa\left(s-t\right)}+\theta\left(1-e^{-\kappa\left(s-t\right)}\right)+\mathcal{O}\left(\nu^{2}+\eta^{2}\right)\right]ds\right)\\
 & =d\left(\nu\bar{\sigma}_{t}^{2}\int_{t}^{T}A(T,s)e^{-\kappa\left(s-t\right)}ds\right)+d\left(\nu\theta\int_{t}^{T}A(T,s)\left(1-e^{-\kappa\left(s-t\right)}\right)ds\right)\\
 & \qquad+d\left(\nu\int_{t}^{T}A(T,s)\mathcal{O}\left(\nu^{2}+\eta^{2}\right)ds\right)\\
 & =\nu d\bar{\sigma}_{t}^{2}\int_{t}^{T}A(T,s)e^{-\kappa\left(s-t\right)}ds+\nu\bar{\sigma}_{t}^{2}d\left(\int_{t}^{T}A(T,s)e^{-\kappa\left(s-t\right)}ds\right)\\
 & \qquad+\nu\theta d\left(\int_{t}^{T}A(T,s)\left(1-e^{-\kappa\left(s-t\right)}\right)ds\right)+\nu\mathcal{O}\left(\nu^{2}+\eta^{2}\right)d\left(\int_{t}^{T}A(T,s)ds\right).
\end{align*}
Applying Leibniz rule to derivate under the integral sign we can rewrite
the previous expression as 
\begin{align*}
dL\left[W,M^{c}\right]_{t} & =\nu d\bar{\sigma}_{t}^{2}\int_{t}^{T}A(T,s)e^{-\kappa\left(s-t\right)}ds\\
 & \qquad+\nu\bar{\sigma}_{t}^{2}\left[-A\left(T,t\right)dt+\kappa\left(\int_{t}^{T}A(T,s)e^{-\kappa\left(s-t\right)}ds\right)dt\right]\\
 & \qquad+\nu\theta\left[-\kappa\left(\int_{t}^{T}A(T,s)e^{-\kappa\left(s-t\right)}ds\right)dt\right]\\
 & \qquad-\nu\mathcal{O}\left(\nu^{2}+\eta^{3}\right)A(T,t)dt\\
 & =\nu^{2}\left(\int_{t}^{T}A(T,s)e^{-\kappa\left(s-t\right)}ds\right)\bar{\sigma}_{t}dW_{t}-\nu\bar{\sigma}_{t}^{2}A\left(T,t\right)dt+\mathcal{O}\left(\nu^{3}+\nu\eta^{2}\right).
\end{align*}
\item In order to see that $dD\left[M^{c},M^{c}\right]_{t}=\nu^{3}\left(\int_{t}^{T}A^{2}(T,s)e^{-\kappa\left(s-t\right)}ds\right)\bar{\sigma}_{t}dW_{t}-\nu^{2}\bar{\sigma}_{t}^{2}A^{2}\left(T,t\right)dt$,
we proceed in an analogous way as in the previous set of computations.
\begin{align*}
dD\left[M^{c},M^{c}\right]_{t} & =d\left(\nu^{2}\int_{t}^{T}A^{2}(T,s)\mathbb{E}_{t}\left[\bar{\sigma}_{s}^{2}\right]ds\right)\\
 & =\nu^{2}d\left(\int_{t}^{T}A^{2}(T,s)\left[\bar{\sigma}_{t}^{2}e^{-\kappa\left(s-t\right)}+\theta\left(1-e^{-\kappa\left(s-t\right)}\right)\right]ds\right)\\
 & =\nu^{2}\left[d\left(\bar{\sigma}_{t}^{2}\int_{t}^{T}A^{2}(T,s)e^{-\kappa\left(s-t\right)}ds\right)+\theta d\left(\int_{t}^{T}A^{2}(T,s)\left(1-e^{-\kappa\left(s-t\right)}\right)ds\right)\right]\\
 & =\nu^{2}\left[d\bar{\sigma}_{t}^{2}\int_{t}^{T}A^{2}(T,s)e^{-\kappa\left(s-t\right)}ds+\bar{\sigma}_{t}^{2}d\left(\int_{t}^{T}A^{2}(T,s)e^{-\kappa\left(s-t\right)}ds\right)\right]\\
 & \qquad+\nu^{2}\theta d\left(\int_{t}^{T}A^{2}(T,s)\left(1-e^{-\kappa\left(s-t\right)}\right)ds\right)\\
 & =\nu^{3}\left(\int_{t}^{T}A^{2}(T,s)e^{-\kappa\left(s-t\right)}ds\right)\bar{\sigma}_{t}dW_{t}-\nu^{2}\bar{\sigma}_{t}^{2}A^{2}\left(T,t\right)dt.
\end{align*}
\end{itemize}
\bibliographystyle{plain}
\bibliography{Bibliography}

\end{document}